\documentclass[submission,copyright,creativecommons]{eptcs}
\providecommand{\event}{MFPS 2021} % Name of the event you are submitting to
\usepackage{breakurl}             % Not needed if you use pdflatex only.
\usepackage[T1]{fontenc}
\usepackage{textcomp}
\usepackage{listings}
\usepackage{graphicx}
\usepackage{amsmath}
\usepackage{bussproofs}
\usepackage{mathtools}
\usepackage{amssymb}
\usepackage{amsthm}
\usepackage{xspace}
\usepackage{ulem}
\usepackage{xcolor}

\usepackage{ifthen}
\usepackage{bm}

\newcommand{\btype}{b}           % base type
\newcommand{\vtype}{\alpha}      % type variable
\newcommand{\ttype}{\sigma}      % top type
\newcommand{\lmtype}{\mtype^{l}}
\newcommand{\rmtype}{\mtype^{r}}
\newcommand{\mtype}{\tau}        % monotype
\newcommand{\xtype}{\chi}        % extensible type
\newcommand{\gkind}{\kappa}      % generic kind restriction
\newcommand{\ukind}{\mathcal{U}} % universal kind restriction
\newcommand{\glab}{l}            % generic label
\newcommand{\lglab}{l^{l}}
\newcommand{\rglab}{l^{r}}

\newcommand{\gterm}{M} % generic term
\newcommand{\vterm}{x} % variable
\newcommand{\cterm}{k} % constant

\newcommand{\geqs}{E}             % generic set of equations to unify
\newcommand{\tenv}{\Gamma}        % typing environment
\newcommand{\kenv}{K}             % kinding environment
\newcommand{\skenv}{S\!K}         % not necessarily well-formed kinding environment
\newcommand{\slab}{\mathcal{L}}   % set of all labels
\newcommand{\norm}{\da}           % norm is down arrow
\newcommand{\normof}[1]{\overline{#1}}           % norm is down arrow
\newcommand{\red}{\Ra_{\xtype}}  % type reduction relation

\newcommand{\subs}{S}
\newcommand{\id}{\textit{id}}
\newcommand{\dom}[1]{\textit{dom}(#1)}
\newcommand{\ftv}[1]{\textit{FTV}(#1)}
\newcommand{\eftv}[2]{\textit{EFTV}(#1,#2)}
\newcommand{\cls}[3]{\textit{Cls}(#1,#2,#3)}
\newcommand{\U}{\mathcal{U}}
\newcommand{\unify}[2]{\U(#1,#2)}               % unification algorithm
\newcommand{\infer}[3]{\textit{WK}(#1,#2,#3)}   % inference algorithm

\newcommand{\base}{\mathcal{B}}
\newcommand{\xbase}[1]{\base(#1)}              % extensible type base type

\newcommand{\atype}[2]{#1 \rightarrow #2}  % arrow type
\newcommand{\extype}[3]{#1 + \{#2 : #3\}}  % type extension
\newcommand{\cntype}[3]{#1 - \{#2 : #3\}}  % type contraction

\newcommand{\kind}[1]{\{\!\{#1\}\!\}}
\newcommand{\rkind}[2]{\{\!\{#1 \mmid #2\}\!\}}

\newcommand{\app}[2]{#1 \ #2}                                               % application
\newcommand{\abs}[2]{\lambda #1.#2}                                         % abstraction
\newcommand{\letin}[3]{\text{let} \ #1 = #2 \ \text{in} \ #3}               % let
\newcommand{\sel}[2]{#1.#2}                                                 % selection
\newcommand{\modif}[3]{\text{modify}(#1,#2,#3)}                             % modify
\newcommand{\ext}[3]{\text{extend}(#1,#2,#3)}                               % record extension
\newcommand{\cnt}[2]{#1 \lminus #2}                                         % record contraction
\newcommand{\Ra}{\Rightarrow}

\newcommand{\da}{\downarrow}
\newcommand{\fields}{\mathit{F}}
\newcommand{\lfields}{\fields^{l}}
\newcommand{\rfields}{\fields^{r}}
\newcommand{\efields}[1]{\fields_{e(#1)}}
\newcommand{\cfields}[1]{\fields_{c(#1)}}
\newcommand{\lminus}{\mathbin{{\setminus}\mspace{-5mu}{\setminus}}}
\newcommand{\mmid}{\mid\!\mid}

\newtheorem{theorem}{Theorem}[section]
\newtheorem{lemma}[theorem]{Lemma}
\newtheorem{definition}[theorem]{Definition}
\newtheorem{example}[theorem]{Example}
\newtheorem{proposition}[theorem]{Proposition}
\newtheorem{corollary}[theorem]{Corollary}

\title{An ML-style Record Calculus with Extensible Records\thanks{This work is financed by National Funds through the Portuguese funding agency, FCT - Fundação para a Ciência e a Tecnologia, within project UIDB/50014/2020.}}
\author{Sandra Alves
\institute{ DCC-FCUP \& CRACS \\ University of Porto, Porto, Portugal
}
\email{sandra@fc.up.pt}
\and
Miguel Ramos
\institute{DCC-FCUP \& LIACC\\ University of Porto, Porto, Portugal}
\email{jmiguelsramos@gmail.com}
}
\def\titlerunning{An ML-style Record Calculus with Extensible Records}
\def\authorrunning{S. Alves \& M. Ramos}
\begin{document}
\maketitle

\begin{abstract} 
In this work, we develop a polymorphic record calculus with extensible records. Extensible records are records that can have new fields added to them, or preexisting fields removed from them. We also develop a static type system for this calculus and a sound and complete type inference algorithm. Most ML-style polymorphic record calculi that support extensible records are based on row variables. We present an alternative construction based on the polymorphic record calculus developed by Ohori. Ohori based his polymorphic record calculus on the idea of kind restrictions. This allowed him to express polymorphic operations on records such as field selection and modification. With the addition of extensible types, we were able to extend Ohori's original calculus with other powerful operations on records such as field addition and removal.

\end{abstract}

\section{Introduction}
A record is a basic data structure that provides a flexible way to aggregate data and which is present in several programming and specification languages.
Record polymorphism has been studied using approaches based on subtyping~\cite{CardelliW85,Cardelli90}, kinded quantification~\cite{Ohori95,OhoriB88} and, most commonly, on the mechanism of row variables~\cite{Wand87,Jategaonkar93,Remy92,Wand89}, among others.
Row variables range over finite sets of field types, which are constructed by extension starting from the empty row $\{||\}$. 
With the aim of developing a sound polymorphic programming language that supported labelled records and variants, while providing efficient compilation, Ohori~\cite{Ohori95} followed an approach based on the notion of a kind. In Ohori's system, variables ranging over record types, are annotated with a specification that represents the fields the record is expected to contain. This refines ML-type quantification to what is called kinded quantification, of the form $\forall \vtype :: \gkind.\ttype$, where type variable $\vtype$ is constrained to range only over the set of types denoted by kind $\gkind$. A kind $\gkind$ is either the universal kind $\ukind$, denoting the set of all types, or a record kind of the form $\kind{\glab_1 : \mtype_1, \dots, \glab_n : \mtype_n}$, denoting the set of all record types that contain fields $\glab_1, \dots, \glab_n$, with types $\mtype_1, \dots, \mtype_n$, respectively.

The type inference algorithm in~\cite{Ohori95} provides a sound extension of ML's let-polymorphism~\cite{DamasM82}, which allows for a polymorphic treatment of record-based operations such as field selection and modification, but with the limitation of lacking support for extensible records. This limitation is often accepted in practical implementations of languages with record types, in a trade for efficiency, or due to the difficulty in guaranteeing the correctness of types for more flexible operations.

Nevertheless, record types, and in particular polymorphic record types, are most relevant not only in the context of ML-style programming languages, but also due to its relevance to areas were data aggregation and manipulation is a key feature.  The style of record polymorphism developed by Ohori was recently explored in the context of event processing, through the development of a domain specific higher-order functional language for events~\cite{AlvesFR20}, with a typing system that was both a restriction and an extension of Ohori’s polymorphic record calculus. Although this domain specific language proved to be adequate to deal with the notion of generic events, with the potential of providing a formal semantics to Complex Event Processing (CEP) systems, the lack of support for extensibility was once more a limitation, since the ability to extend a record with a new field or remove an existing field from a record is often useful in the context of CEP.

In this paper we address this limitation and develop a polymorphic record calculus with extensible records, that is, records that can have new fields added to them, or preexisting fields removed from them. To that end, we refine the notion of record kind, such that, a kind is of the form $$\rkind{\glab_1 : \mtype_1, \dots, \glab_n : \mtype_n}{\glab'_1 : \mtype_1', \dots, \glab'_m : \mtype'_m}$$ and denotes the set of all record types that contain the fields before $\mmid$ and do not contain the fields after $\mmid$. This system allows us to represent polymorphic versions of various types of record operations, such as field selection and modification, just as before, but is also powerful enough to represent field extension (the operation that adds a new field to a record) and field removal (the operation that removes a preexisting field from a record). The notion of record types is also extended to accommodate the notion of extensible record types. We refine the notion of kinded restrictions, using the refined notion of record kind, so that we can impose conditions on the extension and removal operations: a record can only be extended with a field that it does not already contain, and one can only remove existing fields from records. We extend Ohori's calculus with two new operations on records: $\cnt{\gterm}{\glab}$ removes field $\glab$ from term $\gterm$, provided that $\gterm$ is a record with that field, and $\ext{\gterm_1}{\glab}{\gterm_2}$ extends term $\gterm_1$ with a field labelled $\glab$ with value $\gterm_2$, provided that $\gterm_1$ is a record not containing $\glab$.  These restrictions are imposed by the type system. We present a sound and complete ML-style type inference algorithm for our extended calculus. 

The main contributions of this paper are:
\begin{itemize}
\item The design definition of an ML-style record calculus with operations for field extension and field removal.
\item An ML-style type system for this calculus based on the notion of extensible types.
\item A sound and complete type inference algorithm, that extends the ML-style record calculus in~\cite{Ohori95}.
\end{itemize} 
\paragraph{Overview} In Section~\ref{sec:calculus} we define our ML-style record calculus with extensible records. In Section~\ref{sec:ta} we define a type system for our calculus and in Section~\ref{sec:ti} we present a type inference algorithm, which is proved to be sound and complete. We discuss related work in Section~\ref{sec:rw} and we finally conclude and discuss further work in Section~\ref{sec:conc}.

\section{An ML-style Record Calculus with Extensible Records}
\label{sec:calculus}
In this section we introduce an ML-style record calculus with extensible records.  Our set of terms follows the one used by Ohori in~\cite{Ohori95}, except for the exclusion of variants and the addition of two new terms: one for adding a new field to a record; and another for removing a preexisting field from a record. We assume some familiarity with the $\lambda$-calculus (see~\cite{Barendregt85} for a detailed reference).

\subsection{Terms}
We start by formally defining the set of terms. In the following, let $k$ range over a countable set of constants, $x,y,z,\dots$ range over a countable set of variables and $\glab,\glab_1,\dots$ range over a countable set $\slab$ of labels. Additionally we assume a finite set of base types $\mathbb{B}$, ranged by $\btype$.

\begin{definition}
 The set of terms is given by the following grammar:
        \[
                \begin{array}{lcl}
                        \gterm & ::= & \vterm \mid \cterm^{\btype} \mid \abs{\vterm}{\gterm} \mid \gterm \gterm \mid \letin{\vterm}{\gterm}{\gterm} \\
                               &     & \{\glab = \gterm, \dots, \glab = \gterm\} \mid \gterm.\glab \mid \modif{\gterm}{\glab}{\gterm} \mid  \cnt{\gterm}{\glab} \mid \ext{\gterm}{\glab}{\gterm}
                \end{array}
        \]
       % where $\cterm^{\btype}$ is a constant of base type $\btype$.
\end{definition}

\subsection{Types and Kinds}
We now define the set of types and kinds. Following Damas and Milner's presentation of ML~\cite{DamasM82}, we divide the set of types into monotypes (ranged over by $\mtype$) and polytypes (ranged over by $\ttype$). Monotypes can be base types (represented by $\btype$, which is obtained from a given set of base types), extensible types (ranged over by $\xtype$), and arrow types (of the form $\atype{\mtype}{\mtype}$). Polytypes can be monotypes, or quantified types of the form $\forall \vtype::\gkind.\ttype$, where $\gkind$ is a kind restriction and $\vtype$ is a type variable quantified over the set of types denoted by $\gkind$. Extensible types can be type variables (represented by $\vtype$, which is obtained from a given countably infinite set of type variables), record types (of the form $\{\glab : \mtype, \dots, \glab : \mtype\}$, where $\glab$ is obtained from a given set of labels), type extensions (of the form $\extype{\xtype}{\glab}{\mtype}$), or type contractions (of the form $\cntype{\xtype}{\glab}{\mtype}$). Type extensions of the form $\extype{\xtype}{\glab}{\mtype}$ are the type of records of type $\xtype$ that are extended with a new field with label $\glab$ and type $\mtype$ and type contractions of the form $\cntype{\xtype}{\glab}{\mtype}$ are the type of records of type $\xtype$ that have a preexisting field with label $\glab$ and type $\mtype$ removed.
\begin{definition}The set of types $\ttype$ and kinds $\gkind$ are specified by the following grammar:
%        The set of types and kinds is given by the following grammar:
        %
        \begin{align*}
                \ttype & ::= \mtype \mid  \forall \vtype::\gkind.\ttype \\
                \mtype & ::= \btype \mid \xtype \mid \atype{\mtype}{\mtype} \\
                \xtype & ::= \vtype \mid \{\glab:\mtype, \dots, \glab:\mtype\} \mid \extype{\xtype}{\glab}{\mtype} \mid \cntype{\xtype}{\glab}{\mtype}\\
                \gkind & ::= \ukind \mid \rkind{\glab : \ttype, \dots, \glab : \ttype}{\glab : \ttype, \dots, \glab : \ttype}
        \end{align*}

\end{definition}
 All empty records are typed with the same type: the empty record type $\{\}$. Note that a type variable kinded with the universal kind can be instantiated with any type, while a type variable kinded with the empty kind restriction $\rkind{}{}$ can only be instantiated with a record type. Also, note that extensible types are defined recursively and can either have type variables or records types as base cases. To a type that appears as the base case of an extensible type we will call the base type of the extensible type (or just base type, if the context makes this clear). Also, if $\xtype$ is an extensible type, then $\xbase{\xtype}$ is its base type.
 
 Labels that appear in types and kinds must always be pairwise distinct and each label can only be assigned one type during its existence. Also, the order in which labels occur is insignificant. Finally, if two extensible types have the same base type and have the same kind restrictions, we will consider them equal. This will made precise in Section~\ref{sec:ti}, when we introduce a type reduction mechanism for extensible types.

\begin{example}
    \label{ex:equalityofexttype}
    The following extensible types are equal:
    \begin{align}
        \cntype{\extype{\vtype}{\glab_1}{\mtype_1}}{\glab_2}{\mtype_2} & \equiv \extype{\cntype{\vtype}{\glab_2}{\mtype_2}}{\glab_1}{\mtype_1} \\
        \cntype{\cntype{\extype{\vtype}{\glab_1}{\mtype_1}}{\glab_2}{\mtype_2}}{\glab_1}{\mtype_1} & \equiv \cntype{\cntype{\extype{\vtype}{\glab_1}{\mtype_1}}{\glab_1}{\mtype_1}}{\glab_2}{\mtype_2}
    \end{align}
    \begin{itemize}
        \item[(1)] Clearly, despite the order of their type extensions and contractions, both types have the same base type and tell us that records with type $\vtype$ must not have the field $\{\glab_1 : \mtype_1\}$ and must have the field $\{\glab_2 : \mtype_2\}$.
        \item[(2)] Here, we know that both types are equal because, despite the order of their type extensions and contractions, both have the same base type and tell us that records with type $\vtype$ must have the field $\{\glab_2 : \mtype_2\}$ and must not have the field $\{\glab_1 : \mtype_1\}$.
    \end{itemize}
    
    The two following extensible types are not equal:
    \begin{align}
        \cntype{\extype{\vtype_1}{\glab_1}{\mtype_1}}{\glab_2}{\mtype_2} & \not\equiv \extype{\cntype{\vtype_2}{\glab_1}{\mtype_1}}{\glab_2}{\mtype_2} \\
        \extype{\extype{\cntype{\vtype_1}{\glab_1}{\mtype_1}}{\glab_2}{\mtype_2}}{\glab_1}{\mtype_1} & \not\equiv \extype{\cntype{\extype{\vtype_2}{\glab_1}{\mtype_1}}{\glab_1}{\mtype_1}}{\glab_2}{\mtype_2}
    \end{align}
    \begin{itemize}
        \item[(3)] Clearly, these two types are not equal, and not just because they have syntactically different base types, but because the type on the left tells us that $\vtype_1$ must not have the field $\{\glab_1 : \mtype_1\}$ and must have the field $\{\glab_2 : \mtype_2\}$, while the type on the right tells us that $\vtype_2$ must have the field $\{\glab_1 : \mtype_1\}$ and must not have the field $\{\glab_2 : \mtype_2\}$.
        \item[(4)] Here, we also have two types that are not equal, because they must have two syntactically different base types. The reason is that the type on the left tells us that its base type must have the field $\{\glab_1 : \mtype_1\}$ and the type on the right tells us that its base type must not have that same field.
    \end{itemize}
\end{example}

Note that we can only ignore the order of type extensions and contractions with \emph{different labels} since changing the ordering between a type extensions and type contraction of a type will render it ``ill'' formed, as can be seen in Example~\ref{ex:equalityofexttype}. This property is enforced by the kinding rules in Definition~\ref{def:kindrules}.

\begin{definition}
    The set of free type variables of a type $\ttype$ or a kind $\gkind$ are denoted by $\ftv{\ttype}$ and $\ftv{\gkind}$, respectively. For second-order types, it is defined as $\ftv{\forall \vtype :: \gkind.\ttype} = \ftv{\gkind} \cup (\ftv{\ttype} \setminus \{\vtype\})$. \textit{FTV} for other types and kinds are defined as expected.
\end{definition}

We say that a type $\ttype$ is closed if $\ftv{\ttype} = \emptyset$. We assume that all bounded type variables are distinct and different from any free type variables, and that this property is preserved by substitution through $\alpha$-equivalence. The type construct $\forall \vtype :: \gkind.\ttype$ binds the type variable $\vtype$ in $\ttype$, but not in $\gkind$.

\begin{definition} 
%A kind assignment $\kenv$ is a mapping from a finite set of type variables to kinds that sometimes we can regard as set of pairs of a type variable and a kind. 
We write $\{\vtype_1 :: \gkind_1, \dots, \vtype_n :: \gkind_n\}$ for the kind assignment that binds $\vtype_i$ to $\gkind_i$, $(1 \leq i \leq n)$ and $\emptyset$ for the empty kind assignment. We will also write $\kenv \{\vtype :: \gkind\}$ for $\kenv \cup \{\vtype :: \gkind\}$ provided that $\kenv$ is well formed, $\vtype \not\in \dom{\kenv}$, and $\ftv{\gkind} \subseteq \dom{\kenv}$. Note that $\kenv \{\vtype_1 :: \gkind_1, \dots, \vtype_n :: \gkind_n\}$ implies that $\vtype_i \not\in \ftv{\gkind_j}$ for any $1 \leq j < i \leq n$.
\end{definition}

Any type variables that appear in $\kenv$ must also be properly kinded by $\kenv$ itself.

\begin{definition}\label{def:wfkindassign}
        A kind assignment $\kenv$ is well formed if for all $\vtype \in \dom{\kenv}, \ftv{\kenv(\vtype)} \subseteq \dom{\kenv}$, where $\dom{f}$ denotes the domain of a function $f$.
\end{definition}

From now on, we assume that every kind assignment is well formed.

\begin{definition}\label{def:wftuk}
        A type $\ttype$ is well formed under a kind assignment $\kenv$ if $\ftv{\ttype} \subseteq \dom{\kenv}$.
\end{definition}

Definition~\ref{def:wftuk} is naturally extended to other syntactic constructs containing types except substitutions (see Definition~\ref{def:wfsubs}).

\subsection{Kind Restrictions}
\begin{definition}\label{def:kindrules}
        %Let $\{\glab_1 : \mtype_1, \dots, \glab_n : \mtype_n, [\glab'_1 : \mtype'_1, \dots, \glab'_m : \mtype'_m]\} = \{\glab_1 : \mtype_1, \dots, \glab_n : \mtype_n\} \cup A$, where $A \subseteq \{\glab'_1 : \mtype'_1, \dots, \glab'_m : \mtype'_m\}$. 
        Type $\mtype$ has kind restriction $\gkind$, that we denote as $\kenv \Vdash \mtype :: \gkind$, if it can be derived from the following kinding rules:
        \begin{align*}
                i) \ \kenv & \Vdash \mtype :: \ukind \ \text{for any} \ \mtype \ \text{well formed under} \ \kenv \\
                ii) \ \kenv & \Vdash \{\lglab_1 : \lmtype_1, \dots, \lglab_n : \lmtype_n, \dots\} :: \rkind{\lglab_1 : \lmtype_1, \dots, \lglab_n : \lmtype_n}{\rglab_1 : \rmtype_1, \dots, \rglab_m : \rmtype_m} \\
                & \qquad \text{if} \ \{\lglab_1, \dots, \lglab_n, \dots\} \cap \{\rglab_1, \dots, \rglab_m\} = \emptyset, \\ 
                & \qquad \quad \text{both }\{\lglab_1 : \lmtype_1, \dots, \lglab_n : \lmtype_n, \dots\} %\ \text{is well formed under} \ \kenv \\
                %& \qquad \quad 
                \ \ \text{and} \ \rmtype_i \ (1 \leq i \leq m) \ \text{are well formed under} \ \kenv \\
                iii) \ \kenv & \Vdash \vtype :: \rkind{\lglab_1 : \lmtype_1, \dots, \lglab_n : \lmtype_n}{\rglab_1 : \rmtype_1, \dots, \rglab_m : \rmtype_m} \\
                & \qquad \text{if} \ \kenv(\vtype) = \rkind{\lglab_1 : \lmtype_1, \dots, \lglab_n : \lmtype_n, \dots}{\rglab_1 : \rmtype_1, \dots, \rglab_m : \rmtype_m, \dots} \\
                iv) \ \kenv & \Vdash \extype{\xtype}{\glab}{\mtype} :: \rkind{\lglab_1 : \lmtype_1, \dots, \lglab_n : \lmtype_n, [\glab : \mtype]}{\rglab_1 : \rmtype_1, \dots, \rglab_m : \rmtype_m} \\
                & \qquad \text{if} \ \kenv \Vdash \xtype ::  \rkind{\lglab_1 : \lmtype_1, \dots, \lglab_n : \lmtype_n}{\rglab_1 : \rmtype_1, \dots, \rglab_m : \rmtype_m, \glab : \mtype} \\
                v) \ \kenv & \Vdash \cntype{\xtype}{\glab}{\mtype} :: \rkind{\lglab_1 : \lmtype_1, \dots, \lglab_n : \lmtype_n}{\rglab_1 : \rmtype_1, \dots, \rglab_m : \rmtype_m, [\glab : \mtype]} \\
                & \qquad \text{if} \ \kenv \Vdash \xtype ::  \rkind{\lglab_1 : \lmtype_1, \dots, \lglab_n : \lmtype_n, \glab : \mtype}{\rglab_1 : \rmtype_1, \dots, \rglab_m : \rmtype_m}
        \end{align*}
        where $[\glab : \tau]$ means that the inclusion of $\glab : \tau$ in its respective kind is optional.
\end{definition}

Note that if $\kenv \Vdash \mtype : \gkind$, then both $\gkind$ and $\mtype$ are well formed under $\kenv$.     

\begin{example}
    Consider the type $\mtype = \cntype{\extype{\{\glab_1 : \mtype_1\}}{\glab_2}{\mtype_2}}{\glab_1}{\mtype_1}$ and the empty kind assignment $\kenv = \emptyset$. These are some possible derivations of kind restrictions for $\mtype$ and its subtypes:
    \begin{align*}
 %       & \emptyset \Vdash \{\glab_1 : \mtype_1\} :: \ukind \\
        & \emptyset \Vdash \{\glab_1 : \mtype_1\} :: \rkind{}{} \\
        & \emptyset \Vdash \{\glab_1 : \mtype_1\} :: \rkind{\glab_1 : \mtype_1}{} \\
        & \emptyset \Vdash \extype{\{\glab_1 : \mtype_1\}}{\glab_2}{\mtype_2} :: \rkind{\glab_2 : \mtype_2}{} \\
        & \emptyset \Vdash \extype{\{\glab_1 : \mtype_1\}}{\glab_2}{\mtype_2} :: \rkind{}{\glab_3 : \mtype_3} \\
        & \emptyset \Vdash \cntype{\extype{\{\glab_1 : \mtype_1\}}{\glab_2}{\mtype_2}}{\glab_1}{\mtype_1} :: \rkind{\glab_2 : \mtype_2}{\glab_1 : \mtype_1}
    \end{align*}
    Note that we cannot derive any kind restrictions for $\mtype$ where $\{\glab_1 : \mtype_1\}$ appears on the left of the $\mmid$.
\end{example}

\begin{proposition}\label{prop:kindrules}
    The kinding rules ensure the two following properties:
    \begin{enumerate}
        \item If $\kenv \Vdash \mtype :: \rkind{\glab : \mtype', \dots}{\dots}$, then $\kenv \not\Vdash \mtype :: \rkind{\dots}{\glab : \mtype', \dots}$;
        \item If $\kenv \Vdash \mtype :: \rkind{\dots}{\glab : \mtype', \dots}$, then $\kenv \not\Vdash \mtype :: \rkind{\glab : \mtype', \dots}{\dots}$.
    \end{enumerate}
\end{proposition}

% \begin{example}
%         This is an example for Definition~\ref{def:kindrules}.
% \end{example}

\subsection{Kinded Substitutions}
\begin{definition}\label{def:typesubs}
        A type substitution is a function from a finite set of type variables to types. We write $[\ttype_1/\vtype_1, \dots, \ttype_n/\vtype_n]$ for the substitution that maps each $\vtype_i$ to $\ttype_i$. A substitution $\subs$ is extended to the set of all type variables by letting $\subs(\vtype) = \vtype$ for all $\vtype \not\in \dom{\subs}$, where $\dom{f}$ denotes the domain of a function $f$, and is extended uniquely to polytypes, record types, function types and extensible types as follows:
        \begin{align*}
                \subs(\forall\vtype :: \gkind.\ttype) & = \forall\vtype :: \subs(\gkind).\subs(\ttype) \\
                \subs(\{\glab_1 : \mtype_1, \dots, \glab_n : \mtype_n\}) & =  \{\glab_1 : \subs(\mtype_1), \dots, \glab_n : \subs(\mtype_n)\}) \\
                \subs(\atype{\mtype_1}{\mtype_2})                        & = \atype{\subs(\mtype_1)}{\subs(\mtype_2)}                          \\
                \subs(\cntype{\xtype}{\glab}{\mtype})                    & = \cntype{\subs(\xtype)}{\glab}{\subs(\mtype)}                      \\
                \subs(\extype{\xtype}{\glab}{\mtype})                    & = \extype{\subs(\xtype)}{\glab}{\subs(\mtype)}\
        \end{align*}
\end{definition}

\begin{definition}\label{def:wfsubs}
        A substitution $\subs$ is well formed under a kind assignment $\kenv$ if for any $\vtype \in \dom{\subs}$, $\subs(\vtype)$ is well formed under $\kenv$.
\end{definition}

Since all type variables are kinded by a kind assignment, the conventional notion of substitution (see Definition~\ref{def:typesubs}) must be refined by also taking into consideration kind constraints.

\begin{definition}
        A kinded substitution is a pair $(\kenv, \subs)$ of a kind assignment $\kenv$ and a substitution $\subs$ that is well formed under $\kenv$. A kinded substitution $(\kenv, \subs)$ is ground if $\kenv = \emptyset$. We will write $\subs$ for a ground kinded substitution $(\emptyset, \subs)$.
\end{definition}

\begin{definition}\label{def:ksuskenv}
        A kinded substitution $(\kenv_1, \subs)$ respects a kind assignment $\kenv_2$ if for any $\vtype \in \dom{\kenv_2}$, $\kenv_1 \Vdash \subs(\vtype) :: \subs(\kenv_2(\vtype))$.
\end{definition}

This notion specifies the condition under which a substitution can be applied, i.e., if $(\kenv_1, \subs)$ respects $\kenv$ then it can be applied to a type $\ttype$ kinded by $\kenv$, yielding a type $\subs(\ttype)$ kinded by $\kenv_1$. 

\begin{lemma}\label{lem:kindsubs}
    If $\kenv \Vdash \mtype :: \gkind$ and a kinded substitution $(\kenv_1, \subs)$ respects $\kenv$, then $\kenv_1 \Vdash \subs(\mtype) :: \subs(\gkind)$.
\end{lemma}
% \begin{proof}
%     The proof can be found in Appendix~\ref{app:lem:kindsubs}.
% \end{proof}

\begin{corollary}\label{cor:kindsubs}
    If $(\kenv_1, \subs_1)$ respects $\kenv$ and $(\kenv_2, \subs_2)$ respects $\kenv_1$, then $(\kenv_2, \subs_2 \circ \subs_1)$ respects $\kenv$.
\end{corollary}

\begin{proof}
    Since $(\kenv_1, \subs_1)$ respects $\kenv$, we know that $\forall\vtype \in \dom{\kenv}, \kenv_1 \Vdash \subs_1(\vtype) :: \subs_1(\kenv(\vtype))$. Since $(\kenv_2, \subs_2)$ respects $\kenv_1$, we also know that $\forall\vtype \in \dom{\kenv_1}, \kenv_2 \Vdash \subs_2(\vtype) :: \subs_2(\kenv_1(\vtype))$. But, since we know that $(\kenv_2, \subs_2)$ respects $\kenv_1$, we have that $\kenv_2 \Vdash \subs_2 \circ \subs_1(\vtype) :: \subs_2 \circ \subs_1(\kenv(\vtype))$. Therefore, $\forall \vtype \in \dom{\kenv}, \kenv_2 \Vdash \subs_2 \circ \subs_1(\vtype) :: \subs_2 \circ \subs_1(\kenv(\vtype))$ and $(\kenv_2, \subs_2 \circ \subs_1)$ respects $\kenv$.
\end{proof}

% \begin{example}
%         This is an example of Definition~\ref{def:ksuskenv}.
% \end{example}

A kind assignment is a constraint on possible substitutions of type variables. Furthermore, since types may depend on type variables other than their own free type variables, we need to take into consideration kind assignments and extend the notion of free type variables of a type.

\begin{definition}\label{def:eftv}
        For a type $\ttype$ well formed under $\kenv$, the set of essentially-free type variables of $\ttype$ under $\kenv$, denoted $\eftv{\kenv}{\ttype}$, is the smallest set satisfying: 
        \begin{enumerate}
            \item [$(i)$] $\ftv{\ttype} \subseteq \eftv{\kenv}{\ttype}$. 
            \item [$(ii)$] $\mathit{If\ \vtype \in \eftv{K}{\ttype}, then\ \ftv{K(\vtype)} \subseteq \eftv{K}{\ttype}}$.
        \end{enumerate}
        %
        % \begin{itemize}
        %         \item $\ftv{\ttype} \subseteq \eftv{\kenv}{\ttype}$
        %         \item $\mathit{If\ \vtype \in \eftv{K}{\ttype}, then\ \ftv{K(\vtype)} \subseteq \eftv{K}{\ttype}}$
        % \end{itemize}
\end{definition}

% \begin{example}
%         This is an example of Definition~\ref{def:eftv}.
% \end{example}

\section{Type System}
\label{sec:ta}
Before going into the typing rules, we first need to define what are generic instance of monotypes and closures of polytypes. To define the closure of a type, we first need to define what is a type assignment.

\begin{definition}
%A type assignment $\tenv$ is a mapping from a finite set of variables to polymorphic types (polytypes) that we sometimes can regard them as a set of pairs of a term variable and a type. 
We write $\{\vterm_1 : \ttype_1, \dots, \vterm_n : \ttype_n\}$ for the type assignment that binds $\vterm_i$ to $\ttype_i$, $(1 \leq i \leq n)$ and $\emptyset$ for the empty type assignment. We will also write $\tenv \{\vterm :   \ttype\}$ for $\tenv \cup \{\vterm : \ttype\}$ provided that $\vterm \not\in \dom{\tenv}$.
\end{definition}

\begin{definition}
        We say that a type assignment $\tenv$ is well formed under a kind assignment $\kenv$, if $\forall \vterm \in \dom{\tenv}$, $\tenv(\vterm)$ is well formed under $\kenv$.
\end{definition}

\begin{definition}\label{def:typecls}
        Let $\tenv$ and $\mtype$ be well formed under $\kenv$. The closure of $\mtype$ under $\tenv$, $\kenv$, denoted by $\cls{\kenv}{\tenv}{\mtype}$, is a pair $(\kenv', \forall \vtype_1 :: \gkind_1 \cdots \forall \vtype_n :: \gkind_n.\mtype)$ such that $\kenv' \{\vtype_1 :: \gkind_1, \dots, \vtype_n :: \gkind_n \} = \kenv$ and $\{\vtype_1, \dots, \vtype_n\} = \eftv{\kenv}{\mtype} \setminus \eftv{\kenv}{\tenv}$.
\end{definition}

\begin{definition}\label{def:typeinst}
        Let $\ttype_1$ be a polytype well formed under $\kenv$. We say that $\ttype_2$ is a generic instance of $\ttype_1$ under $\kenv$, written $\kenv \Vdash \ttype_1 \geq \ttype_2$, if $\ttype_1 = \forall \vtype^{1}_{1} :: \gkind^{1}_{1} \cdots \vtype^{1}_{n} :: \gkind^{1}_{n}.\mtype_1$, $\ttype_2 = \forall \vtype^{2}_{1} :: \gkind^{2}_{1} \cdots \vtype^{2}_{1} :: \gkind^{2}_{m}.\mtype_2$, and there is a substitution $\subs$ such that $\dom{\subs} = \{\vtype^{1}_{1}, \dots, \vtype^{1}_{n}\}$, $(\kenv \{\vtype^{2}_{1} :: \gkind^{2}_{1}, \dots, \vtype^{2}_{m} :: \gkind^{2}_{m}\}, \subs)$ respects $\kenv \{\vtype^{1}_{1} :: \gkind^{1}_{1}, \dots, \vtype^{1}_{n} :: \gkind^{1}_{n}\}$ and $\mtype_2 = \subs(\mtype_1)$.
\end{definition}

\begin{lemma}\label{lem:transinst}
    If $\kenv \Vdash \ttype_1 \geq \ttype_2$ and $\kenv \Vdash \ttype_2 \geq \ttype_3$, then $\kenv \Vdash \ttype_1 \geq \ttype_3$.
\end{lemma}

\begin{proof}
    Without loss of generality, let us assume that $\ttype_1 = \forall \vtype^{1}_{1} :: \gkind^{1}_{1} \cdots \vtype^{1}_{n} :: \gkind^{1}_{n}.\mtype_1$, $\ttype_2 = \forall \vtype^{2}_{1} :: \gkind^{2}_{1} \cdots \vtype^{2}_{m} :: \gkind^{1}_{m}.\mtype_2$, and $\ttype_3 = \forall \vtype^{3}_{1} :: \gkind^{3}_{1} \cdots \vtype^{3}_{k} :: \gkind^{3}_{k}.\mtype_3$. Since $\kenv \Vdash \ttype_1 \geq \ttype_2$, we know that there exists a substitution $\subs_1$ such that $\dom{\subs_1} = \{\vtype^1_{1}, \dots, \vtype^1_{n}\}$, $(\kenv\{\vtype^2_{1} :: \gkind^2_{1}, \dots, \vtype^2_{m} :: \gkind^2_{m}\}, \subs_1)$ respects $\kenv\{\vtype^1_{1} :: \gkind^1_{n} :: \gkind^1_{n}\}$ and $\mtype_2 = \subs_1(\mtype_1)$. Since $\kenv \Vdash \ttype_2 \geq \ttype_3$, we know that there exists a substitution $\subs_2$ such that $\dom{\subs_2} = \{\vtype^2_{1}, \dots, \vtype^2_{m}\}$, $(\kenv\{\vtype^3_{1} :: \gkind^3_{1}, \dots, \vtype^3_{k} :: \gkind^3_{k}\}, \subs_2)$ respects $\kenv\{\vtype^2_{1} :: \gkind^2_{n} :: \gkind^2_{m}\}$ and $\mtype_3 = \subs_2(\mtype_2)$. To show that $\kenv \Vdash \ttype_1 \geq \ttype_3$, we just need to find a substitution $\subs_3$, such that $\dom{\subs_3} = \{\vtype^1_{1}, \dots, \vtype^1_{n}\}$, $(\kenv\{\vtype^3_{1} :: \gkind^3_{1}, \dots, \vtype^3_{k} :: \gkind^3_{k}\}, \subs_3)$ respects $\kenv\{\vtype^1_{1} :: \gkind^1_{n} :: \gkind^1_{n}\}$ and $\mtype_3 = \subs_3(\mtype_1)$. If we choose $\subs_3 = \subs_2 \circ \subs_1$, then $\dom{\subs_3} = \dom{\subs_1} = \{\vtype^1_{1}, \dots, \vtype^1_{n}\}$, by Corollary~\ref{cor:kindsubs}, $(\kenv\{\vtype^3_{1} :: \gkind^3_{1}, \dots, \vtype^3_{k} :: \gkind^3_{k}\}, \subs_2 \circ \subs_1)$ respects $\kenv$, and $\mtype_3 = \subs_3(\mtype_1) = \subs_2 \circ \subs_1(\mtype_1) = \subs_2(\subs_1(\mtype_1)) = \subs_2(\mtype_2)$.
\end{proof}

The type assignment system is given in Figure~\ref{fig:typesystem}. We use $\kenv,\tenv \vdash \gterm: \sigma$ to denote that term $\gterm$ has type $\sigma$ given the type and kind assignments $\tenv$ and $\kenv$, respectively.

\begin{figure}
        \begin{prooftree}
                \AxiomC{$\tenv \ \text{is well formed under} \ \kenv$}
                \AxiomC{$\kenv \Vdash \tenv(\vterm) \geq \mtype$}
                \RightLabel{(Var)}
                \BinaryInfC{$\kenv, \tenv \vdash \vterm : \mtype$}
        \end{prooftree}
        \begin{prooftree}
                \AxiomC{$\tenv \ \text{is well formed under} \ \kenv$}
                \RightLabel{(Const)}
                \UnaryInfC{$\kenv, \tenv \vdash \cterm^\btype : \btype$}
        \end{prooftree}
        \begin{prooftree}
                \AxiomC{$\kenv, \tenv \{\vterm : \mtype_1\} \vdash \gterm : \mtype_2$}
                \RightLabel{(Abs)}
                \UnaryInfC{$\kenv, \tenv \vdash \abs{\vterm}{\gterm} : \atype{\mtype_1}{\mtype_2}$}
        \end{prooftree}
        \begin{prooftree}
                \AxiomC{$\kenv, \tenv \vdash \gterm_1 : \atype{\mtype_1}{\mtype_2}$}
                \AxiomC{$\kenv, \tenv \vdash \gterm_2 : \mtype_1$}
                \RightLabel{(App)}
                \BinaryInfC{$\kenv, \tenv \vdash \app{\gterm_1}{\gterm_2} : \mtype_2$}
        \end{prooftree}
        \begin{prooftree}
                \AxiomC{$\kenv, \tenv \vdash \gterm_1 : \ttype$}
                \AxiomC{$\kenv, \tenv \{\vterm : \ttype\} \vdash \gterm_2 : \mtype$}
                \RightLabel{(Let)}
                \BinaryInfC{$\kenv, \tenv \vdash \letin{\vterm}{\gterm_1}{\gterm_2} : \mtype$}
        \end{prooftree}
        \begin{prooftree}
                \AxiomC{$\kenv, \tenv \vdash \gterm_i : \mtype_i \ (1 \le i \le n)$}
                \RightLabel{(Rec)}
                \UnaryInfC{$\kenv, \tenv \vdash \{l_1 = \gterm_1, \dots, l_n = \gterm_n\} : \{l_1 : \mtype_1, \dots, l_n : \mtype_n\}$}
        \end{prooftree}
        \begin{prooftree}
                \AxiomC{$\kenv, \tenv \vdash \gterm : \mtype_1$}
                \AxiomC{$\kenv \Vdash \mtype_1 :: \rkind{\glab : \mtype_2}{}$}
                \RightLabel{(Sel)}
                \BinaryInfC{$\kenv, \tenv \vdash \sel{M}{\glab} : \mtype_2$}
        \end{prooftree}
        \begin{prooftree}
                \AxiomC{$\kenv, \tenv \vdash \gterm_1 : \mtype_1$}
                \AxiomC{$\kenv, \tenv \vdash \gterm_2 : \mtype_2$}
                \AxiomC{$\kenv \Vdash \mtype_1 :: \rkind{\glab : \mtype_2}{}$}
                \RightLabel{(Modif)}
                \TrinaryInfC{$\kenv, \tenv \vdash \modif{\gterm_1}{\glab}{\gterm_2} : \mtype_1$}
        \end{prooftree}
        \begin{prooftree}
                \AxiomC{$\kenv, \tenv \vdash \gterm : \mtype$}
                \AxiomC{$\cls{\kenv}{\tenv}{\mtype} = (\kenv', \ttype)$}
                \RightLabel{(Gen)}
                \BinaryInfC{$\kenv', \tenv \vdash \gterm : \ttype$}
        \end{prooftree}
        \begin{prooftree}
                \AxiomC{$\kenv, \tenv \vdash \gterm : \mtype_1$}
                \AxiomC{$\kenv \Vdash \mtype_1 :: \rkind{\glab : \mtype_2}{}$}
                \RightLabel{(Contr)}
                \BinaryInfC{$\kenv, \tenv \vdash \cnt{\gterm}{\glab} : \cntype{\mtype_1}{\glab}{\mtype_2}$}
        \end{prooftree}
        \begin{prooftree}
                \AxiomC{$\kenv, \tenv \vdash \gterm_1 : \mtype_1$}
                \AxiomC{$\kenv, \tenv \vdash \gterm_2 : \mtype_2$}
                \AxiomC{$\kenv \Vdash \mtype_1 :: \rkind{}{\glab : \mtype_2}$}
                \AxiomC{$\xbase{\mtype_1} \not\in \ftv{\mtype_2}$}
                \RightLabel{(Ext)}
                \QuaternaryInfC{$\kenv, \tenv \vdash \ext{\gterm_1}{\glab}{\gterm_2} : \extype{\mtype_1}{\glab}{\mtype_2}$}
        \end{prooftree}
        \caption{Typing rules for ML-style record calculus with extensible records}
\label{fig:typesystem}
\end{figure}

\begin{example}
    Let $\kenv = \{\vtype_1 :: \rkind{}{\glab : \vtype_2}, \vtype_2 :: \ukind\}$ and $\tenv = \{\vterm : \vtype_1, y : \vtype_2\}$. Then we can construct the following type derivation for $\sel{\ext{\vterm}{\glab}{y}}{\glab}$:
    \begin{prooftree}
        \AxiomC{$\kenv \Vdash \tenv(\vterm) \geq \vtype_1$}
        \RightLabel{(Var)}
        \UnaryInfC{$\kenv, \tenv \vdash \vterm : \vtype_1$}
        \AxiomC{$\kenv \Vdash \tenv(y) \geq \vtype_2$}
        \RightLabel{(Var)}
        \UnaryInfC{$\kenv, \tenv \vdash y : \vtype_2$}
        \AxiomC{$\kenv \Vdash \vtype_1 :: \rkind{}{\glab : \vtype_2}$}
        \AxiomC{$\vtype_1 \not\in \ftv{\vtype_2}$}
        \RightLabel{(Ext)}
        \LeftLabel{$\Delta =$}
        \QuaternaryInfC{$\kenv, \tenv \vdash \ext{\vterm}{\glab}{y} : \extype{\vtype_1}{\glab}{\vtype_2}$}
    \end{prooftree}
    \begin{prooftree}
        \AxiomC{$\Delta$}
        \AxiomC{$\kenv \Vdash \extype{\vtype_1}{\glab}{\vtype_2} :: \rkind{\glab : \vtype_2}{}$}
        \RightLabel{(Sel)}
        \BinaryInfC{$\kenv, \tenv \vdash \sel{\ext{\vterm}{\glab}{y}}{\glab} : \vtype_2$}        
    \end{prooftree}
\end{example}

%In this system, polymorphic generalization and let abstraction are separated into two rules: Gen and Let. It is possible to combine the two into a single rule, but it makes it harder to prove various properties that can be easily proved by induction on typing derivations. 
The following lemma allows us to strengthen the type assignment.

\begin{lemma}\label{lem:eqinst}
    If $\kenv, \tenv \{\vterm : \ttype_1\} \vdash \gterm : \mtype$ and $\kenv \Vdash \ttype_2 \geq \ttype_1$, then $\kenv, \tenv \{\vterm : \ttype_2\} \vdash \gterm : \mtype$.
\end{lemma}

% \begin{proof}
% The proof can be found in Appendix~\ref{app:lem:eqinst}.
% \end{proof}

The following lemma shows that typings are closed under kind-respecting kinded substitutions.

\begin{lemma}\label{lem:closedksubs}
    If $\kenv_1, \tenv \vdash \gterm : \ttype$ and $(\kenv_2, \subs)$ respects $\kenv_1$, then $\kenv_2, \subs(\tenv) \vdash \gterm : \subs(\ttype)$.
\end{lemma}

% \begin{proof}
% The proof can be found in Appendix~\ref{app:lem:closedksubs}.
% \end{proof}

\section{Type Inference}
\label{sec:ti}
If we add a field to a record and then immediately remove it, we end up with a record with the same set of fields and the same kind restrictions, but with a different type. The same thing happens if we remove a preexisting field from a record and then immediately add it again. This means that extensible types can have different forms, but represent the same set of record types and records. This induces a set of identities on extensible types that can be interpreted as rewrite rules and used to define a type reduction system.

\begin{definition}
    \label{def:rewriterules}Let $\xtype$ be well formed under some kind assignment $\kenv$. The rewrite rules for the type reduction system depend on the form of $\xtype$ and are the following:
    \begin{align*}
         i) & \ \{\glab_1 : \mtype_1, \dots, \glab_i : \mtype_i, \dots, \glab_n : \mtype_n\} - \{\glab_i : \mtype_i\} \pm \{\glab' : \mtype'\} \cdots\ \ \red \ \ \{\glab_1 : \mtype_1, \dots, \glab_n : \mtype_n\} \pm \{\glab' : \mtype'\} \cdots \\\\
         ii) & \ \{\glab_1 : \mtype_1, \dots, \glab_n : \mtype_n\} + \{\glab : \mtype\} \pm \{\glab' : \mtype'\} \cdots\ \ \red \ \ \{\glab_1 : \mtype_1, \dots, \glab_n : \mtype_n, \glab : \mtype\} \pm \{\glab' : \mtype'\} \cdots \\\\
         iii) & \ \vtype \pm_1 \{\glab_1 : \mtype_1\} \cdots -_i \{\glab : \mtype\} \cdots +_j \{\glab : \mtype\} \cdots\ \  \red \\
         & \ \vtype \pm_1 \cdots \pm_{i-1} \{\glab_{i-1} : \mtype_{i~+1}\} \pm_{i+1} \{\glab_{i+1} : \mtype_{i+1}\} \cdots \pm_{j-1} \{\glab_{j-1} : \mtype_{j-1}\} \pm_{j+1} \{\glab_{j+1} : \mtype_{j+1}\} \cdots \\ \\
         iv) & \ \vtype \pm_1 \{\glab_1 : \mtype_1\} \cdots +_i \{\glab : \mtype\} \cdots -_j \{\glab : \mtype\} \cdots \ \ \red \\
         & \ \vtype \pm_1 \cdots \pm_{i-1} \{\glab_{i-1} : \mtype_{i-1}\} \pm_{i+1} \{\glab_{i+1} : \mtype_{i+1}\} \cdots \pm_{j-1} \{\glab_{j-1} : \mtype_{j-1}\} \pm_{j+1} \{\glab_{j+1} : \mtype_{j+1}\} \cdots
    \end{align*}
    To ensure that $\red$ is confluent without considering the order of type extensions and/or contractions, we are going to assume that, in rules \textit{iii)} and \textit{iv)}, $i$ and $j$ are the positions of the two first occurrences of those same-labelled type contraction and extension, in the case of rule \textit{iii)}, or type extension and contraction, in the case of rule \textit{ix)}.
\end{definition}

We can show that the type reduction system is convergent, i.e. that normal forms are unique. This fact in conjunction with the fact that each reduction step preserves the set of kind restrictions of a type (see Proposition~\ref{prop:preserve}), allows us to replace any type with its normal form. Having types in their normal forms allows us to to identify types that correspond to records that are constructed from a particular record by adding and removing fields in different orders but that end up having the same set of fields.

\begin{proposition}\label{prop:convergence}
    $\red$ is convergent.
\end{proposition}

It is obvious that $\red$ is terminating, since the number of type extensions and type contractions diminishes with each reduction step. Also, note that $\red$ is confluent by construction.

\begin{definition}\label{def:normform}Let $\red^*$ be the reflexive and transitive closure of $\red$. $(i)$ We say that $\ttype_2$ is a normal form of $\ttype_1$ if, and only if, $\ttype_1 \red^* \ttype_2$ and there is no $\ttype_3$ such that $\ttype_2 \red \ttype_3$. $(ii)$ We write $\normof{\ttype_1}$ for the (unique) normal form of $\ttype_1$. $(iii)$ We write $\ttype_1 \norm \ttype_2$ if $\normof{\ttype_1} = \normof{\ttype_2}$.
    %     \item We say that $\ttype_2$ is a normal form of $\ttype_1$ if, and only if, $\ttype_1 \red^* \ttype_2$.
    %     \item If $\ttype_1$ has a uniquely determined normal form, we will write denote it as $\normof{\ttype_1}$.
    % \end{enumerate}
\end{definition}
    
%\begin{definition}\label{def:joinable}
%     We say that $\ttype_1$ and $\ttype_2$ are ``joinable'' if, and only if, there is a $\ttype_3$ such that $\ttype_1 \red^* \ttype_3$ and $\ttype_2 \red^* \ttype_3$, in which case we write $\ttype_1 \norm \ttype_2$.
% \end{definition}
 
Now we show that each reduction step $\red$ preserves kind restrictions.

\begin{proposition}\label{prop:preserve}
    Every rewrite rule $\red$, transforms a type $\xtype_1$ into a $\xtype_2$, such that $\kenv \Vdash \xtype_1 :: \gkind$ if, and only if, $\kenv \Vdash \xtype_2 :: \gkind$.
\end{proposition}

Every field that appears in an extensible type that is a subtype of a type in normal form occurs exactly once.

\begin{proposition}\label{prop:canonicity}
    If $\xtype$ is in normal form, then every label that appears in it occurs exactly once.
\end{proposition}

Note that if a type is in normal form, then every extensible type that appears as its subtype will have a type variable base type and each field that appears in it will occur exactly once.

\begin{proposition}\label{prop:iden_norm}
    Let $\subs$ be a substitution and $\mtype$ be a type. Then $\normof{\subs(\normof{\mtype})} = \normof{\subs(\mtype)}$.
\end{proposition}

\begin{proof}
    Let $\normof{\mtype} = \mtype$. Then there exists an extensible type $\xtype$ in $\mtype$ that can be normalized to its normal form $\normof{\xtype}$.
    Let $\vtype \in \dom{\subs}$ appear in $\xtype$, but not in $\normof{\xtype}$. then $\vtype$ only appears in type operations of $\xtype$ that were removed by normalization.
    Clearly, $\subs(\vtype)$ will appear in $\subs(\xtype)$, but not in $\subs(\normof{\xtype})$. But then, since $\vtype$ only appeared in type operations of $\xtype$ that were removed by normalization, $\subs(\vtype)$ will only appear in those same type operations. Therefore, $\subs(\vtype)$ will not appear in $\normof{\subs(\xtype)}$. Thus $\subs(\vtype)$ will neither appear in $\normof{\subs(\normof{\xtype})}$ nor in $\normof{\subs(\xtype)}$.
    Let $\vtype \in \dom{\subs}$ appear in $\xtype$ and in $\normof{\xtype}$. Then $\vtype$ must not appear in any type operations of $\xtype$ that can be removed by normalization.
    Then $\subs(\vtype)$ will appear both in $\subs(\xtype)$ and $\subs(\normof{\xtype})$:
    \begin{itemize}
        \item If $\subs(\normof{\xtype})$ is in normal form, then $\subs(\vtype)$ is in normal form and will appear in both $\normof{\subs(\xtype)}$ and $\normof{\subs(\normof{\xtype})}$;
        \item If $\subs(\normof{\xtype})$ is not in normal form, then either $\subs(\vtype)$ is not in normal form, $\subs(\vtype)$ is in normal form, but type operations that can be cancelled out by preexisting type operations in $\normof{\xtype}$ were added to $\normof{\xtype}$, or both. In any case, these operations will not appear in $\normof{\subs(\normof{\xtype})}$ or $\normof{\subs(\xtype)}$. 
    \end{itemize}
    This means we can conclude that type type operations that appear in $\normof{\subs(\normof{\xtype})}$ and $\normof{\subs(\xtype)}$ are precisely the same, \textit{i.e.} $\normof{\subs(\normof{\xtype})} = \normof{\subs(\xtype)}$.
\end{proof}

To account for normal forms, we need to redefine what it means for two substitutions to be equal.

\begin{definition}\label{def:eqsubs}
    Let $\subs_1$ and $\subs_2$ be two substitutions. We say that $\subs_1 = \subs_2$ if, for any type $\mtype$, $\subs_1(\mtype) \norm \subs_2(\mtype)$.
\end{definition}

Note that if two substitutions are syntactically equal, then they are also equal in the sense of Definition~\ref{def:eqsubs}.

\begin{proposition}
    Let $\subs_1$, $\subs_2$, $\subs_1'$, and $\subs_2'$ be substitutions, such that $\subs_1 = \subs_2$ and $\subs_1' = \subs_2'$. Then $\subs_1 \circ \subs_1' = \subs_2 \circ \subs_2'$.
\end{proposition}

\begin{proof}Since $\subs_1 = \subs_2$ and $\subs_1' = \subs_2'$, we know that $\subs_1(\tau) = \subs_2(\tau)$ and $\subs_1'(\tau) = \subs_2'(\tau)$ for any type $\tau$. But then $\normof{\subs_1'(\tau)} = \normof{\subs_2'(\tau)}$. Therefore, $\subs_1(\normof{\subs_1'(\tau)}) = \subs_2(\normof{\subs_2'(\tau)})$ and, by Proposition~\ref{prop:iden_norm}, $\subs_1(\subs'_1(\tau)) = \subs_2 (\subs'_2(\tau))$. Thus, $\subs_1 \circ \subs'_1 = \subs_2 \circ \subs'_2$.
\end{proof}

\subsection{Kinded Unification}
We are going to extend the kinded unification algorithm used by Ohori in~\cite{Ohori95} with rules for unifying extensible types with type variables and other extensible types. % To unify an extensible type with a type variable, we need to check if the kind restrictions of its base type, in conjunction with its type extensions and contractions, are consistent with the kind restriction of the type variable; to unify an extensible type with another extensible type, we first need to unify the types of the fields of their sequences of type extensions and/or contractions, and then unify the remaining types.

\begin{definition}
        A kinded set of equations is a pair $(\kenv, \geqs)$ consisting of a kind assignment $\kenv$ and a set $\geqs$ of pairs of types such that $\geqs$ is well formed under $\kenv$.
\end{definition}

\begin{definition}
        A substitution $\subs$ satisfies $\geqs$ if $\subs(\mtype_1) \norm \subs(\mtype_2)$, for all $(\mtype_1, \mtype_2) \in \geqs$.
\end{definition}

\begin{definition}
        A kinded substitution $(\kenv_1, \subs)$ is a unifier of a kinded set of equations $(\kenv, \geqs)$ if it respects $\kenv$ and if $\subs$ satisfies $\geqs$.
\end{definition}

\begin{definition}
        $(\kenv_1, \subs)$ is a most general unifier of $(\kenv_2, \geqs)$ if it is a unifier of $(\kenv_2, \geqs)$ and if for any unifier $(\kenv_3, \subs_2)$ of $(\kenv_2, \geqs)$, there is some substitution $\subs_3$ such that $(\kenv_3, \subs_3)$ respects $\kenv_1$ and $\subs_2 = \subs_3 \circ \subs$.
\end{definition}

The unification algorithm $\U$ is defined by transformation. Each rule transforms a 4-tuple of the form $(\geqs, \kenv, \subs, \skenv)$ consisting of a set $\geqs$ of type equations, a kind assignment $\kenv$, a substitution $\subs$, and a (not necessarily well-formed) kind assignment $\skenv$. $\geqs$ keeps the set of equations to be unified; $\kenv$ specifies kind constraints to be verified; $\subs$ records ``solved'' equations in the form of a substitution; and $\skenv$ records ``solved'' kind constraints that have already been verified for $\subs$.

In specifying rules, we treat $\kenv$, $\skenv$, and $\subs$ as sets of pairs. We also use the following notations. Let $\fields$ range over functions from a finite set of labels to types. We write $\{\fields\}$ and $\rkind{\lfields}{\rfields}$ to denote the record type identified by $\fields$ and the record kind identified by $\lfields$ and $\rfields$, respectively. For two functions $\fields_1$ and $\fields_2$ we write $\fields_1 + \fields_2$ for the function $\fields$ such that $\dom{\fields} = \dom{\fields_1} \cup \dom{\fields_2}$ and, for $\glab \in \dom{\fields}$, if $\glab \in \dom{\fields_1}$, then $\fields(\glab) = \fields_1(\glab)$; otherwise, $\fields(\glab) = \fields_2(\glab)$. For those same two functions, we write $\fields_1 - \fields_2$ for the function $\fields$ such that $\dom{\fields} = \dom{\fields_1} \setminus \dom{\fields_2}$ and, for $\glab \in \dom{\fields}$, $\fields(\glab) = \fields_1(\glab)$. For an extensible type $\xtype$, we write $\efields{\xtype}$ and $\cfields{\xtype}$ for the functions that represent the set of type extensions and contractions of $\xtype$, respectively.

\begin{example}
    Let $\xtype = \extype{\cntype{\extype{\vtype}{\glab_1}{\mtype_1}}{\glab_2}{\mtype_2}}{\glab_3}{\mtype_3}$. Then $\efields{\xtype}$ is the function (of domain $\{\glab_1, \glab_3\}$) that sends $\glab_1$ to $\mtype_1$ and $\glab_3$ to $\mtype_3$ and $\cfields{\xtype}$ is the function (of domain $\{\glab_2\}$) that sends $\glab_2$ to $\mtype_2$.
\end{example}

\begin{definition}\label{def:unifalg}
    Let $(\kenv, \geqs)$ be a given kinded set of equations. Algorithm $\U$ first transforms the system $(\geqs, \kenv, \emptyset, \emptyset)$ into $(\geqs', \kenv', \subs, \skenv)$ until no more rules can be applied. It then returns the pair $(\kenv', \subs)$, if $\geqs' = \emptyset$; otherwise, it fails. Its rules can be found in Figure~\ref{fig:kindunif}.
\end{definition}

\begin{figure}
    {\small
    \begin{align*}
             i) & \ (\geqs \cup \{(\mtype_1, \mtype_2)\}, \kenv, \subs, \skenv) \Ra (\geqs, \kenv, \subs, \skenv) \ \text{if} \ \mtype_1 \norm \mtype_2 \\
             ii) & \ (\geqs \cup \{(\vtype, \mtype)\}, \kenv \cup \{(\vtype, \ukind)\}, \subs, \skenv) \Ra ([\mtype/\vtype] \geqs, [\mtype/\vtype] \kenv, [\mtype/\vtype](\subs) \cup \{(\vtype, \mtype)\}, [\mtype/\vtype](\skenv) \cup \{(\vtype, \ukind)\}) \\
             & \qquad \text{if} \ \vtype \not\in \ftv{\mtype} \\ 
             iii) & \ (\geqs \cup \{(\vtype_1, \vtype_2)\}, \kenv \cup \{(\vtype_1, \rkind{\lfields_1}{\rfields_1}), (\vtype_2, \rkind{\lfields_2}{\rfields_2})\}, \subs, \skenv) \Ra \\
             & \qquad ([\vtype_2/\vtype_1](\geqs \cup \{(\lfields_1(\glab), \lfields_2(\glab)) \mid \glab \in \dom{\lfields_1} \cap \dom{\lfields_2}\} \cup \{(\rfields_1(\glab), \rfields_2(\glab)) \mid \glab \in \dom{\rfields_1} \cap \dom{\rfields_2}\}), \\
             & \qquad \ [\vtype_2/\vtype_1](\kenv) \cup \{(\vtype_2, [\vtype_2/\vtype_1](\rkind{\lfields_1 + \lfields_2}{\rfields_1 + \rfields_2}))\}, \\
             & \qquad \ [\vtype_2/\vtype_1](\subs)\cup \{(\vtype_1, \vtype_2)\}, [\vtype_2/\vtype_1](\skenv) \cup \{(\vtype_1, \rkind{\lfields_1}{\rfields_1})\}) \\
             & \qquad \text{if} \ \dom{\lfields_1} \cap \dom{\rfields_2} = \emptyset \ \text{and} \ \dom{\rfields_1} \cap \dom{\lfields_2} = \emptyset \\
             iv) & \ (\geqs \cup \{(\vtype, \{F_2\})\}, \kenv \cup \{(\vtype, \rkind{\lfields_1}{\rfields_1})\}, \subs, \skenv) \Ra \\
             & \qquad ([\{\fields_2\}/\vtype](\geqs \cup \{(\lfields_1(l), \fields_2(l)) \mid \glab \in \dom{\lfields_1}\}), \\
             & \qquad \ [\{\fields_2\}/\vtype](\kenv),  [\{\fields_2\}/\vtype](\subs) \cup \{(\vtype, \{\fields_2)\})\}, [\{\fields_2\}/\vtype](\skenv) \cup \{(\vtype, \rkind{\lfields_1}{\rfields_1})\}) \\
             & \qquad \text{if} \ \dom{\lfields_1} \subseteq \dom{\fields_2}, \dom{\rfields_1} \cap \dom{\fields_2} = \emptyset, \ \text{and} \ \vtype \not \in \ftv{\{F_2\}} \\
             v) & \ (\geqs \cup \{(\{F_1\}, \{F_2\})\}, \kenv, \subs, \skenv) \Ra (\geqs \cup \{(F_1(l), F_2(l)) \mid l \in \dom{F_1}\}, \kenv, \subs, \skenv) \\
             & \qquad \text{if} \ \dom{F_1} = \dom{F_2} \\
             vi) & \ (\geqs \cup \{(\atype{\mtype^1_1}{\mtype^2_1}, \atype{\mtype^1_2}{\mtype^2_2}\}, \kenv, \subs, \skenv) \Ra (\geqs \cup \{(\mtype^1_1, \mtype^1_2), (\mtype^2_1, \mtype^2_2)\}, \kenv, \subs, \skenv) \\
             vii) & \ (\geqs \cup \{(\vtype, \xtype)\}, \kenv \cup \{(\vtype, \rkind{\lfields_1}{\rfields_1}), (\xbase{\xtype}, \rkind{\lfields_2}{\rfields_2})\}, \subs, \skenv) \Ra \\
             & \qquad ([\xtype/\vtype](\geqs \cup \{(\lfields_1(\glab), (\rfields_2 + (\lfields_2 - \cfields{\normof{\xtype}}))(\glab)) \mid \glab \in \dom{\lfields_1} \cap \dom{\rfields_2 + (\lfields_2 - \cfields{\normof{\xtype}})}\} \\
             & \qquad \qquad \qquad \cup \{(\rfields_1(\glab), \cfields{\normof{\xtype}}(\glab)) \mid \glab \in \dom{\rfields_1} \cap \dom{\cfields{\normof{\xtype}}}\}), \\
             & \qquad \ [\xtype/\vtype](\kenv) \cup \{(\xbase{\xtype}, [\xtype/\vtype](\rkind{\lfields_2 + (\lfields_1 - (\rfields_2 + (\lfields_2 - \cfields{\normof{\xtype}})))}{\rfields_2 + (\rfields_1 - \cfields{\normof{\xtype}})}))\}, \\
             & \qquad \ [\xtype/\vtype](\subs)\cup \{(\vtype, \xtype)\}, [\xtype/\vtype](\skenv) \cup \{(\vtype, \rkind{\lfields_1}{\rfields_1})\}) \\
             & \qquad \text{if} \ \dom{\lfields_1} \cap \dom{\cfields{\normof{\xtype}}} = \emptyset, \dom{\rfields_1} \cap \dom{\rfields_2 + (\lfields_2 - \cfields{\normof{\xtype}})} = \emptyset , \ \text{and} \ \vtype \not\in \ftv{\xtype} \\
             viii) & \ (\geqs \cup \{(\vtype^1 \pm^1_1 \{\glab^1_1 : \mtype^1_1\} \cdots \pm^1_i \{\glab^1_i : \mtype^1_i\} \cdots \pm^1_n \{\glab^1_n : \mtype^1_n\}, \\
             & \qquad \quad \ \ \vtype^2 \pm^2_1 \{\glab^2_1 : \mtype^2_1\} \cdots \pm^2_j \{\glab^2_j : \mtype^2_j\} \cdots \pm^2_m \{\glab^1_m : \mtype^2_m\})\}, \kenv, \subs, \skenv) \Ra \\
             & \qquad (\geqs \cup \{(\mtype^1_i, \mtype^2_j), (\vtype^1 \pm^1_1 \{\glab^1_1 : \mtype^1_1\} \cdots \pm^1_{i-1} \{\glab^1_{i-1} : \mtype^1_{i-1}\} \pm^1_{i+1} \cdots \pm^1_n \{\glab^1_n : \mtype^1_n\}, \\
             & \qquad \qquad \qquad \qquad \quad \vtype^2 \pm^2_1 \{\glab^2_1 : \mtype^2_1\} \cdots \pm^2_{j-1} \{\glab^2_{j-1} : \mtype^2_{j-1}\} \pm^2_{j+1} \cdots \pm^2_m \{\glab^2_m : \mtype^2_m\})\}, \kenv, \subs, \skenv) \\
             & \qquad \text{if} \ (\pm^1_i = \pm^2_j \wedge \glab^1_i = \glab^2_j), \forall i < k \leq n : \glab^1_k \not= \glab^2_i, \ \text{and} \ \forall j < r \leq m : \glab^2_r \not=\glab^2_j \\
             ix) & \ (\geqs \cup \{(\vtype^1 \pm^1_1 \{\glab^1_1 : \mtype^1_1\} \cdots \pm^1_n \{\glab^1_n : \mtype^1_n\}, \vtype^2 \pm^2_1 \{\glab^2_1 : \mtype^2_1\} \cdots \pm^2_m \{\glab^2_m : \mtype^2_m\})\}, \\
             & \qquad \kenv \cup \{(\vtype_1, \rkind{\lfields_1}{\rfields_1}), (\vtype_2, \rkind{\lfields_2}{\rfields_2})\}, \subs, \skenv) \Ra \\
             & \qquad (\subs_{ix)}(\geqs \cup \{(\lfields_1(\glab), \lfields_2(\glab)) \mid \glab \in \dom{\lfields_1} \cap \dom{\lfields_2}\} \cup \{(\rfields_1(\glab), \rfields_2(\glab)) \mid \glab \in \dom{\rfields_1} \cap \dom{\rfields_2}\}), \\
             & \qquad \ \subs_{ix)}(\kenv) \cup \{(\vtype, \subs_{ix)}(\rkind{\lfields_1 + \lfields_2}{\rfields_1 + \rfields_2})\}, \\
             & \qquad \ \subs_{ix)}(\subs) \cup \{(\vtype_1, \vtype \pm^2_1 \{\glab^2_1 : \mtype^2_1\} \cdots \pm^2_m \{\glab^2_m : \mtype^2_m\}), (\vtype_2, \vtype \pm^1_1 \{\glab^1_1 : \mtype^1_1\} \cdots \pm^1_n \{\glab^1_n : \mtype^1_n\})\} \\
             & \qquad \ \subs_{ix)}(\skenv) \cup \{(\vtype_1, \rkind{\lfields_1}{\rfields_1}), (\vtype_2, \rkind{\lfields_2}{\rfields_2})\}) \\ 
             & \qquad \ \text{where} \ \subs_{ix)} = [\vtype \pm^2_1 \{\glab^2_1 : \mtype^2_1\} \cdots \pm^2_m \{\glab^2_m : \mtype^2_m\}/\vtype_1, \vtype \pm^1_1 \{\glab^1_1 : \mtype^1_1\} \cdots \pm^1_n \{\glab^1_n : \mtype^1_n\}/\vtype_2] \\
             & \qquad \text{if} \ \dom{\lfields_1} \cap \dom{\rfields_2} = \emptyset, \dom{\rfields_1} \cap \dom{\lfields_1} = \emptyset, \vtype_1 \not\in \ftv{\vtype^2 \pm^2_1 \{\glab^2_1 : \mtype^2_1\} \cdots \pm^2_m \{\glab^2_m : \mtype^2_m\}}, \\
             & \qquad \quad \vtype_2 \not\in \ftv{\vtype^1 \pm^1_1 \{\glab^1_1 : \mtype^1_1\} \cdots \pm^1_n \{\glab^1_n : \mtype^1_n\}}, \forall 1 \leq i \leq n, 1 \leq j \leq m, \glab^1_i \not= \glab^1_j \ \text{and} \ \vtype \ \text{is fresh}
    \end{align*}}
    \vspace{-0.2in}
    \caption{Kinded Unification}
  \label{fig:kindunif}
\end{figure}
Since the constructions present in rule \textit{vii)} from Definition~\ref{def:unifalg} are somewhat intricate, we will state the following facts in hope that they will help the reader better understand this rule. Let $\vtype$ be type variable and $\xtype$ an extensible type, both well formed under some kind assignment $\kenv$, such that $\xbase{\xtype} = \vtype_{\xtype}$, $\kenv(\vtype) = \rkind{\lfields_1}{\rfields_1}$, $\kenv(\vtype_{\xtype}) = \rkind{\lfields_2}{\rfields_2}$. Also, let us assume that $\dom{\lfields_1} \cap \dom{\cfields{\normof{\xtype}}} = \emptyset$, $\dom{\rfields_1} \cap \dom{\rfields_2 + (\lfields_2 - \cfields{\normof{\xtype}})} = \emptyset$ and $\vtype \not\in \ftv{\xtype}$:
\begin{itemize}
    \item Any field that appears in $\lfields_2$ was either introduced by the (Sel), (Modif) or (Contr) rule. This means that the set of the labels that appear in type contractions of $\xtype$ is $\dom{\cfields{\normof{\xtype}}} \subseteq \dom{\lfields_2}$. 
    \item Any field that appears on $\rfields_2$ was introduced by the (Ext) rule. This means that the set of the labels that appear in type extensions of $\xtype$ is $\dom{\rfields_2} \subseteq \dom{\efields{\normof{\xtype}}}$.
    \item The set of fields that are guaranteed to appear in a record typed with $\xtype$ is represented by the function $\rfields_2 + (\lfields_2 - \cfields{\normof{\xtype}})$ and the set of fields that are guaranteed to \textit{not} appear in a record typed with $\xtype$ is represented by the function $\cfields{\normof{\xtype}}$.
    \item The set of fields that are guaranteed to exist by $\kenv(\vtype)$, but not by $\xtype$, is represented by the function $\lfields_1 - (\rfields_2 + (\lfields_2 - \cfields{\normof{\xtype}})))$ and the set of fields that are guaranteed \textit{not} to exist by $\kenv(\vtype)$, but not by $\xtype$, is represented by the function $\rfields_1 - \cfields{\normof{\xtype}}$.
\end{itemize}
Note that we use $\normof{\xtype}$ instead of $\xtype$ in $\cfields{\normof{\xtype}}$ so that constructing this function is more straightforward. That being said, to keep proofs simpler, we do not normalize extensible types during unification.% because this would complicate the proof of Theorem~\ref{thm:unifalg}.

\begin{theorem}\label{thm:unifalg}
    Algorithm $\U$ takes any kinded set of equations and computes a most general unifier, if one exists; otherwise it fails.
\end{theorem}

% \begin{proof}
%     The complete proof can be found in Appendix~\ref{app:thm:unifalg}.
% \end{proof}

\begin{example}
    Let $\kenv = \{\vtype :: \rkind{}{\glab : \vtype''}, \vtype' :: \rkind{\glab : \vtype''}{}, \vtype'' :: \ukind\}$. Then $\U(\kenv, \{(\cntype{\extype{\vtype}{\glab}{\vtype''}}{\glab}{\vtype''}, \cntype{\vtype'}{\glab}{\vtype''})\} = (\{\vtype :: \rkind{\glab : \vtype'}{}\}, \{(\vtype', \extype{\vtype}{\glab}{\vtype''})\})$:
    \begin{align*}
        & (\{(\cntype{\extype{\vtype}{\glab}{\vtype''}}{\glab}{\vtype''},\cntype{\vtype'}{\glab}{\vtype''})\}, \kenv, \emptyset, \emptyset) \\
        \overset{\textit{viii)}}{\Ra} & (\{(\extype{\vtype}{\glab}{\vtype''}, \vtype')\}, \kenv, \emptyset, \emptyset) \\
        \overset{\textit{vii)}}{\Ra} & (\emptyset, \{\vtype :: \rkind{}{\glab : \vtype''}, \vtype'' :: \ukind\}, \{(\vtype', \extype{\vtype}{\glab}{\vtype''})\}, \{(\vtype', \rkind{\glab : \vtype''}{})\})
    \end{align*}
    Note that the first transformation is given by rule \textit{viii)} and not by rule \textit{ix)}. This is the case because the latter rule can only be applied whenever two extensible types have no matching type extensions or contractions.
\end{example}
It is true that our unification algorithm introduces an overhead associated with the search of matching type extensions and contraction (see rules \textit{viii)} and \textit{ix)}) when compared with Ohori's original algorithm. That being said, we believe that with some care during the algorithm's implementation this overhead should be negligible.
\subsection{Type Inference}
Using the kinded unification, we extend Ohori's type inference algorithm to the cases of record field extension and record field removal. 
The \emph{type inference algorithm}, $\infer{\kenv}{\tenv}{\gterm}$, is defined in Figure~\ref{fig:typeinf}. 
        Given a  kind assignment $\kenv$, a type assignment $\tenv$, and a term $\gterm$, the Type Inference Algorithm $\infer{\kenv}{\tenv}{\gterm}$ returns a tuple $(\kenv', \subs, \ttype)$.  It is implicitly assumed that the algorithm fails if unification or any of the recursive calls fail.

\begin{figure}[htbp] 
    {\small
        \begin{align*}
                 i) & \ \infer{\kenv}{\tenv}{\vterm} =  \text{if} \ \vterm \not \in \dom{\tenv} \ \text{then} \ \textit{fail} \\
                 & \ \quad \qquad \qquad \qquad \text{else let} \ \forall \vtype_1::\gkind_1 \cdots \forall \vtype_n::\gkind_n.\mtype = \tenv(\vterm), \\
                 & \ \quad \quad \qquad \qquad \qquad \qquad \ \subs = [\beta_1/\vtype_1, \dots, \beta_n/\vtype_n] \ (\beta_1, \dots, \beta_n \ \text{are fresh}) \\
                 & \qquad \qquad \qquad \quad \ \text{in} \ (\kenv \{\beta_1::\subs(\gkind_1), \dots, \beta_n::\subs(\gkind_n)\}, \id, \normof{\subs(\mtype)}) \\
                 ii) & \ \infer{\kenv}{\tenv}{\abs{\vterm}{\gterm}} = \text{let} \ (\kenv_1, \subs_1, \mtype_1) = \infer{\kenv \{\vtype::\ukind\}}{\tenv \{\vterm : \vtype\}}{\gterm} \ (\vtype \ \text{fresh}) \\
                 & \qquad \qquad  \qquad \qquad \ \ \ \  \text{in} \ (\kenv_1, \subs_1, \atype{\normof{\subs_1(\vtype)}}{\mtype_1}) \\
                 iii) & \ \infer{\kenv}{\tenv}{\app{\gterm_1}{\gterm_2}} = \text{let} \ (\kenv_1, \subs_1, \mtype_1) = \infer{\kenv}{\tenv}{\gterm_1} \\
                 &  \qquad \qquad  \qquad  \qquad  \qquad  \ \  (\kenv_2, \subs_2, \mtype_2) = \infer{\kenv_1}{\subs_1(\tenv)}{\gterm_2} \\
                 &  \qquad \qquad  \qquad  \qquad  \qquad  \ \ (\kenv_3, \subs_3) = \unify{\kenv_2 \{\vtype :: \ukind\}}{\{(\subs_2(\mtype_1), \atype{\mtype_2}{\vtype})\}} \ (\vtype \ \text{fresh}) \\
                 & \qquad \qquad  \qquad  \qquad  \qquad   \text{in} \ (\kenv_3, \subs_3 \circ \subs_2 \circ \subs_1, \normof{\subs_3(\vtype)}) \\
                 iv) & \ \infer{\kenv}{\tenv}{\letin{x}{\gterm_1}{\gterm_2}} = \text{let} \ (\kenv_1, \subs_1, \mtype_1) = \infer{\kenv}{\tenv}{\gterm_1} \\
                 & \qquad \qquad \qquad \qquad \qquad \qquad \qquad \ \ \ \  (\kenv'_1, \sigma) = \cls{\kenv_1}{\subs_1(\tenv)}{\mtype_1} \\
                 & \qquad \qquad \qquad \qquad \qquad \qquad \qquad \ \  \ \  (\kenv_2, \subs_2, \mtype_2) = \infer{\kenv'_1}{(\subs_1(\tenv))\{\vterm : \ttype\}}{\gterm_2} \\
                 & \qquad \qquad \qquad \qquad \qquad \qquad \qquad \ \  \text{in} \ (\kenv_2, \subs_2 \circ \subs_1, \mtype_2) \\
                 v) & \ \infer{\kenv}{\tenv}{\{l_1 = \gterm_1, \dots, l_n = \gterm_n\}} = \\
                 & \qquad \text{let} \ (\kenv_1, \subs_1, \mtype_1) = \infer{\kenv}{\tenv}{\gterm_1} \\
                 & \qquad \qquad (\kenv_i, \subs_i, \mtype_i) = \infer{\kenv_{i-1}}{\subs_{i-1} \circ \cdots \circ \subs_1 (\tenv)}{\gterm_i} \ (2 \le i \le n) \\
                 & \qquad \text{in} \ (\kenv_n, \subs_n \circ \cdots \circ \subs_2 \circ \subs_1, \{l_1 : \normof{\subs_n \circ \cdots \circ \subs_2(\mtype_1)}, \dots, l_i : \normof{\subs_n \circ \cdots \circ \subs_{i+1}(\mtype_i)}, \dots, l_n : \mtype_n\}) \\
                 vi) & \ \infer{\kenv}{\tenv}{\sel{\gterm}{l}} = \text{let} \ (\kenv_1, \subs_1, \mtype_1) =  \infer{\kenv}{\tenv}{\gterm} \\
                 & \qquad \qquad \qquad \qquad \ \ \ (\kenv_2, \subs_2) = \unify{\kenv_1 \{\vtype_1 :: \ukind, \vtype_2 :: \rkind{l : \vtype_1}{}\}}{\{(\vtype_2, \mtype_1)\}} \ (\vtype_1, \vtype_2 \ \text{fresh}) \\
                 & \qquad \qquad \qquad \qquad \ \text{in} \ (\kenv_2, \subs_2 \circ \subs_1, \normof{\subs_2(\vtype_1)}) \\
                 vii) & \ \infer{\kenv}{\tenv}{\modif{\gterm_1}{l}{\gterm_2}} = \\
                 & \qquad \text{let} \ (\kenv_1, \subs_1, \mtype_1) = \infer{\kenv}{\tenv}{\gterm_1} \\
                 & \qquad \qquad (\kenv_2, \subs_2, \mtype_2) = \infer{\kenv_1}{\subs_1(\tenv)}{\gterm_2} \\
                 & \qquad \qquad (\kenv_3, \subs_3) = \unify{\kenv_2 \{\vtype_1 :: \ukind, \vtype_2 :: \rkind{l : \vtype_1}{}\}}{\{(\vtype_1, \mtype_2), (\vtype_2, \subs_2(\mtype_1))\}} \ (\vtype_1, \vtype_2 \ \text{fresh}) \\
                 & \qquad \text{in} \ (\kenv_3, \subs_3 \circ \subs_2 \circ \subs_1, \normof{\subs_3(\vtype_2)}) \\
                 viii) & \ \infer{\kenv}{\tenv}{\cnt{\gterm}{\glab}} = \text{let} \ (\kenv_1, \subs_1, \mtype_1) = \infer{\kenv}{\tenv}{\gterm} \\
                 & \qquad \qquad  \qquad  \qquad \quad \ \  \  (\kenv_2, \subs_2) = \unify{\kenv_1 \{\vtype_1 :: \ukind, \vtype_2 :: \rkind{\glab : \vtype_1}{}\}}{\{(\vtype_2, \mtype_1)\}} \ (\vtype_1, \vtype_2 \ \text{fresh}) \\
                 & \qquad \qquad  \qquad  \qquad  \quad \  \text{in} \ (\kenv_2, \subs_2 \circ \subs_1, \normof{\subs_2(\cntype{\vtype_2}{\glab}{\vtype_1})}) \\
                 ix) & \ \infer{\kenv}{\tenv}{\ext{\gterm_1}{\glab}{\gterm_2}} = \\
                 & \qquad \text{let} \ (\kenv_1, \subs_1, \mtype_1) = \infer{\kenv}{\tenv}{\gterm_1} \\
                 & \qquad \qquad (\kenv_2, \subs_2, \mtype_2) = \infer{\kenv_1}{\subs_1(\tenv)}{\gterm_2} \\
                 & \qquad \text{in} \ \text{if} \ \xbase{\mtype_1} \in \ftv{\mtype_2} \ \text{then} \ \textit{fail} \\
                 & \quad \qquad \ \text{else let} \ (\kenv_3, \subs_3) = \unify{\kenv_2 \{\vtype_1 :: \ukind, \vtype_2 :: \rkind{}{\glab : \vtype_1}\}}{\{(\vtype_1, \mtype_2), (\vtype_2, \subs_2(\mtype_1))\}} \ (\vtype_1, \vtype_2 \ \text{fresh}) \\
                 & \ \quad \qquad \qquad \text{in} \ (\kenv_3, \subs_3 \circ \subs_2 \circ \subs_1, \normof{\subs_3(\extype{\vtype_2}{\glab}{\vtype_1})}) \\
        \end{align*}}
        \vspace{-0.35in}
   \caption{Type inference algorithm}
   \label{fig:typeinf}
\end{figure}        

\begin{proposition}\label{prop:normalform}
    If $\infer{\kenv}{\tenv}{\gterm} = (\kenv', \subs, \mtype)$, then $\mtype$ is in normal form.
\end{proposition}

\begin{theorem}\label{thm:typeinfalg}
        The Type Inference Algorithm is sound and complete:
        \begin{itemize}
                \item (Soundness) If $\infer{\kenv}{\tenv}{\gterm} = (\kenv', \subs, \mtype)$, then $(\kenv', \subs)$ respects $\kenv$, $\kenv', \subs(\tenv) \vdash \gterm : \mtype$.
                \item (Completeness) 
                \begin{itemize}
                    \item If $\infer{\kenv}{\tenv}{\gterm} = (\kenv', \subs, \mtype)$ and $\kenv_0, \subs_0(\tenv) \vdash \gterm : \mtype_0$, for some $(\kenv_0, \subs_0)$ and $\mtype_0$ such that $(\kenv_0, \subs_0)$ respects $\kenv$, then there is some $\subs'$, such that $(\kenv_0, \subs')$ respects $\kenv'$, $\subs'(\mtype) \da \mtype_0$, and $\subs_0(\tenv) = \subs' \circ \subs(\tenv)$;
                    \item If $\infer{\kenv}{\tenv}{\gterm} = \textit{fail}$, then there is no $(\kenv_0, \subs_0)$ and $\mtype_0$ such that $(\kenv_0, \subs_0)$ respects $\kenv$ and $\kenv_0, \subs_0(\tenv) \vdash \gterm : \mtype_0$. 
                \end{itemize}
        \end{itemize}
\end{theorem}

% \begin{proof}
%     The complete proof can be found in Appendix~\ref{app:thm:typeinfalg}.
% \end{proof}

\begin{example}
    Let $\kenv = \{\vtype_1 :: \rkind{}{\glab : \vtype_2}, \vtype_2 :: \ukind\}$ and $\tenv = \{\vterm : \vtype_1, y : \vtype_2\}$. Then we can apply the type inference algorithm to $\sel{\ext{\vterm}{\glab}{y}}{\glab}$ and get the following results:
    \begin{align*}
        & \infer{\kenv}{\tenv}{\sel{\ext{\vterm}{\glab}{y}}{\glab}} = (\kenv, \{(\vtype_3, \vtype_2), (\vtype_4, \vtype_1), (\vtype_5, \vtype_2), (\vtype_6, \extype{\vtype_1}{\glab}{\vtype_2})\}, \vtype_2) \\
        & \infer{\kenv}{\tenv}{\ext{\vterm}{\glab}{y}} = (\kenv, \{(\vtype_3, \vtype_2), (\vtype_4, \vtype_1)\}, \extype{\vtype_1}{\glab}{\vtype_2}) \\
        & \infer{\kenv}{\tenv}{\vterm} = (\kenv, \id, \vtype_1) \\
        & \infer{\kenv}{\tenv}{y} = (\kenv, \id, \vtype_2)
    \end{align*}
\end{example}
\section{Related Work}
\label{sec:rw}
% There are several alternative type systems that deal with polymorphic records with some form of extensibility in literature, with the most common approaches based on subtyping~\cite{Cardelli90,Jategaonkar93} or row variables~\cite{Wand87,HarperM93,Remy92,Wand89}, but also other approaches based on flags~\cite{Remy89,Remy94}, predicates~\cite{HarperP91,Gaster98}, scope variable~\cite{Leijen05}, to name a few. 

There are several alternative type systems in the literature that deal with polymorphic records with some form of extensibility. The most common approaches are based on subtyping~\cite{Cardelli90,Jategaonkar93} or row variables~\cite{Wand87,HarperM93,Remy92,Wand89}, but there are also others based on flags~\cite{Remy89,Remy94}, predicates~\cite{HarperP91,Gaster98} and scope variables~\cite{Leijen05}, to name a few.

The approaches using subtyping~\cite{Cardelli88,CardelliW85,Jategaonkar93,PT1994} have been widely used to build polymorphic type systems with records, in particular for object-oriented programming languages. However, there are several issues that arise when combining record polymorphism with a strong mechanism of subtyping. In the the presence of a subtyping relation $(r_1 \leq r_2)$, meaning that $r_1$ contains at least the fields in $r_2$, a selection operator $(\_.l)$ can have a type $\forall \alpha.\forall \beta \leq \{l:\alpha\}. \beta \rightarrow \alpha$, meaning that a label $l$ can be selected for a record with type $\beta$, if $\beta$ is a subtype of $\{l:\alpha\}$. This leads to additional information on the remaining fields of  $\beta$ being lost, making it harder to define operations dealing with extensibility. This was overcome by moving to a second-order type system~\cite{Cardelli90,CardelliM91}, however the resulting type system relied on explicitly typed extensible records, yielding a system where type-checking and subtyping are decidable, but type-inference is not addressed.  The existence of a  subtyping relation also complicates compilation, (again) because information on the exact type of a record can be lost, leading to the need of incorporating some degree of dynamic typing at runtime.

Several approaches dealing with extensibility use Wand's notion of row variables~\cite{Wand87}, which are variables that range over sets of field types, allowing for incremental construction of records.  However, unlike the approach followed in our paper, operations in~\cite{Wand87} are unchecked, meaning that, when extending a row with a field  $l$, one does not check if this introduces a new field or replaces an existing one, leading to programs for which a principal type does not exist. Flexible systems with extensible records have been constructed over the mechanism of row variables~\cite{Remy89,Remy94}, extended with the notion of flags, yielding a system with extensible records and principal types. Flags are used to provide information on which fields the record must have, and which it cannot have. However, despite the flexibility to define various powerful operations on records, compilation is not dealt efficiently in the presence of flags, due in part to the ability to support some unchecked operations.

Harper and Pierce~\cite{HarperP91} have studied type systems for extensible records where presence and absence of fields is given by predicates on types, thus leading to a system with checked operations, but without dealing with type inference or compilation. The use of predicates was further developed by Jones in his general theory of qualified types~\cite{Jones94,Jones94a}, where extensible records are presented has a special case. Building on that is the approach by Gaster and Jones~\cite{Gaster96,Gaster98} that combines the notion of row variables with the notion of qualified types, and is perhaps the work that is more closely related to ours. In this approach, row extension is used to capture positive information about the fields, while predicates are used to capture negative information, thus avoiding duplicated labels. In our approach both negative and positive information is given by the kind restrictions, resulting in a type system where constraints on label addition and label removal are treated in a uniform way. 

Building on the work of Wand, Rémy and Gaster and Jones~\cite{Wand87,Remy94,Gaster96}, Leijen has developed a polymorphic type system with extensible records based on scoped labels~\cite{Leijen05}. In this approach, duplicate labels are allowed and retained, and an appropriate scoping mechanism is provided to prevent ambiguity and still allow for safe operations. This provides a notion of free extension where update and extension operations are clearly separated, yielding a system that is flexible from the user's point of view. Our approach does not allow for duplicated labels and uses the kinding restrictions to implement a strict notion of extensibility. 

In addition to record extension, there are several systems that deal with other powerful record operations such as concatenation~\cite{HarperP91,Remy92,Wand89}, or the natural join~\cite{BunemanO96} (an operation largely used in database programming, where labeled records play an important role). We choose not to include these operation as they tend to complicate both the implementation and the typing analysis, and simply follow Ohori's approach~\cite{Ohori95}, which efficiently supports the basic operations on records, extending it with two basic operations that support extensibility. Nevertheless, Rémy~\cite{Remy92} has developed an encoding of record concatenation via record extension, thus proving that a system supporting checked record extension can also support some form of record concatenation.

Our decision to build our work on Ohori's record calculus was highly motivated by the existence of an efficient compilation for such a calculus. This was achieved by translating the polymorphic record calculus into an implementation calculus, in which records are represented as vectors whose elements are accessed by direct indexing based on a total order $\ll$ on the set of labels.  Polymorphic functions containing polymorphic record operations are compiled through the insertion of appropriate index abstractions, indicated by the kinded quantifiers given by the type of the function.
Ohori then shows that this algorithm preserves types and that the compilation calculus has the subject reduction property, thus showing that the compilation algorithm preserves the operational behaviour of the original polymorphic record calculus. We believe that an efficient compilation algorithm can also be defined for the calculus developed in this work because variables still range over complete record types and the negative information that was added to kinds only affects type inference, not compilation.
%, as can be seen in the previous example related to the insertion of appropriate index abstractions. 
For these reasons, it should be possible to extend the compilation method in~\cite{Ohori95} to our extensible records.
%to cover this calculus and also to prove its correctness.
\section{Conclusions and Future Work}
\label{sec:conc}
We have presented an ML-style polymorphic record calculus with extensible records, developed a typing system based on the notion of kinded quantification and a sound and complete type inference algorithm, based on kinded unification. While records are a basic commodity in a variety of programming languages, regarding the use of more powerful operations on records, there is still a gap between the theory and the practice, with no consensus on what is the best approach. With this work, we hope to contribute to that discussion.

Ohori's main goal was to support the most common operations dealing with polymorphic records, while maintaining an efficient compilation method. As already stated in the related work, although we do not deal with compilation in this paper, this is something that we would like to address in future work.

%\nocite{*}
\bibliographystyle{eptcs}
\bibliography{refs}

\begin{thebibliography}{10}
\providecommand{\bibitemdeclare}[2]{}
\providecommand{\surnamestart}{}
\providecommand{\surnameend}{}
\providecommand{\urlprefix}{Available at }
\providecommand{\url}[1]{\texttt{#1}}
\providecommand{\href}[2]{\texttt{#2}}
\providecommand{\urlalt}[2]{\href{#1}{#2}}
\providecommand{\doi}[1]{doi:\urlalt{http://dx.doi.org/#1}{#1}}
\providecommand{\eprint}[1]{arXiv:\urlalt{https://arxiv.org/abs/#1}{#1}}
\providecommand{\bibinfo}[2]{#2}

\bibitemdeclare{inproceedings}{AlvesFR20}
\bibitem{AlvesFR20}
\bibinfo{author}{Sandra \surnamestart Alves\surnameend},
  \bibinfo{author}{Maribel \surnamestart Fern{\'{a}}ndez\surnameend} \&
  \bibinfo{author}{Miguel \surnamestart Ramos\surnameend}
  (\bibinfo{year}{2020}): \bibinfo{title}{{EVL:} {A} Typed Higher-order
  Functional Language for Events}.
\newblock In \bibinfo{editor}{Cl{\'{a}}udia \surnamestart Nalon\surnameend} \&
  \bibinfo{editor}{Giselle \surnamestart Reis\surnameend}, editors: {\sl
  \bibinfo{booktitle}{Proceedings of the 15th International Workshop on Logical
  and Semantic Frameworks with Applications, {LSFA} 2020, Online, September 15,
  2020}}, {\sl \bibinfo{series}{Electronic Notes in Theoretical Computer
  Science}} \bibinfo{volume}{351}, \bibinfo{publisher}{Elsevier}, pp.
  \bibinfo{pages}{3--23}, \doi{10.1016/j.entcs.2020.08.002}.

\bibitemdeclare{book}{Barendregt85}
\bibitem{Barendregt85}
\bibinfo{author}{Hendrik~Pieter \surnamestart Barendregt\surnameend}
  (\bibinfo{year}{1985}): \bibinfo{title}{The lambda calculus - its
  syntax and semantics}.
\newblock {\sl \bibinfo{series}{Studies in logic and the foundations of
  mathematics}} \bibinfo{volume}{103}, \bibinfo{publisher}{North-Holland}.

\bibitemdeclare{article}{BunemanO96}
\bibitem{BunemanO96}
\bibinfo{author}{Peter \surnamestart Buneman\surnameend} \&
  \bibinfo{author}{Atsushi \surnamestart Ohori\surnameend}
  (\bibinfo{year}{1996}): \bibinfo{title}{Polymorphism and Type Inference
  in Database Programming}.
\newblock {\sl \bibinfo{journal}{{ACM} Trans. Database Syst.}}
  \bibinfo{volume}{21}(\bibinfo{number}{1}), pp. \bibinfo{pages}{30--76},
  \doi{10.1145/227604.227609}.

\bibitemdeclare{article}{Cardelli88}
\bibitem{Cardelli88}
\bibinfo{author}{Luca \surnamestart Cardelli\surnameend}
  (\bibinfo{year}{1988}): \bibinfo{title}{A Semantics of Multiple
  Inheritance}.
\newblock {\sl \bibinfo{journal}{Inf. Comput.}}
  \bibinfo{volume}{76}(\bibinfo{number}{2/3}), pp. \bibinfo{pages}{138--164},
  \doi{10.1016/0890-5401(88)90007-7}.

\bibitemdeclare{inproceedings}{Cardelli90}
\bibitem{Cardelli90}
\bibinfo{author}{Luca \surnamestart Cardelli\surnameend} \&
  \bibinfo{author}{John~C. \surnamestart Mitchell\surnameend}
  (\bibinfo{year}{1990}): \bibinfo{title}{Operations on records}.
\newblock In \bibinfo{editor}{M.~\surnamestart Main\surnameend},
  \bibinfo{editor}{A.~\surnamestart Melton\surnameend},
  \bibinfo{editor}{M.~\surnamestart Mislove\surnameend} \&
  \bibinfo{editor}{D.~\surnamestart Schmidt\surnameend}, editors: {\sl
  \bibinfo{booktitle}{Mathematical Foundations of Programming Semantics}},
  \bibinfo{publisher}{Springer New York}, \bibinfo{address}{New York, NY}, pp.
  \bibinfo{pages}{22--52}, \doi{10.1145/62678.62700}.

\bibitemdeclare{article}{CardelliM91}
\bibitem{CardelliM91}
\bibinfo{author}{Luca \surnamestart Cardelli\surnameend} \&
  \bibinfo{author}{John~C. \surnamestart Mitchell\surnameend}
  (\bibinfo{year}{1991}): \bibinfo{title}{Operations on Records}.
\newblock {\sl \bibinfo{journal}{Math. Struct. Comput. Sci.}}
  \bibinfo{volume}{1}(\bibinfo{number}{1}), pp. \bibinfo{pages}{3--48},
  \doi{10.1017/S0960129500000049}.

\bibitemdeclare{article}{CardelliW85}
\bibitem{CardelliW85}
\bibinfo{author}{Luca \surnamestart Cardelli\surnameend} \&
  \bibinfo{author}{Peter \surnamestart Wegner\surnameend}
  (\bibinfo{year}{1985}): \bibinfo{title}{On Understanding Types, Data
  Abstraction, and Polymorphism}.
\newblock {\sl \bibinfo{journal}{{ACM} Comput. Surv.}}
  \bibinfo{volume}{17}(\bibinfo{number}{4}), pp. \bibinfo{pages}{471--522},
  \doi{10.1145/6041.6042}.

\bibitemdeclare{inproceedings}{DamasM82}
\bibitem{DamasM82}
\bibinfo{author}{Lu{\'{\i}}s \surnamestart Damas\surnameend} \&
  \bibinfo{author}{Robin \surnamestart Milner\surnameend}
  (\bibinfo{year}{1982}): \bibinfo{title}{Principal Type-Schemes for
  Functional Programs}.
\newblock In \bibinfo{editor}{Richard~A. \surnamestart DeMillo\surnameend},
  editor: {\sl \bibinfo{booktitle}{Conference Record of the Ninth Annual {ACM}
  Symposium on Principles of Programming Languages, Albuquerque, New Mexico,
  USA, January 1982}}, \bibinfo{publisher}{{ACM} Press}, pp.
  \bibinfo{pages}{207--212}, \doi{10.1145/582153.582176}.

\bibitemdeclare{phdthesis}{Gaster98}
\bibitem{Gaster98}
\bibinfo{author}{Benedict~R. \surnamestart Gaster\surnameend}
  (\bibinfo{year}{1998}): \bibinfo{title}{Records, variants and qualified
  types}.
\newblock Ph.D. thesis, \bibinfo{school}{University of Nottingham, {UK}}.
\newblock
  \urlprefix\url{http://ethos.bl.uk/OrderDetails.do?uin=uk.bl.ethos.262959}.

\bibitemdeclare{techreport}{Gaster96}
\bibitem{Gaster96}
\bibinfo{author}{Benedict~R. \surnamestart Gaster\surnameend} \&
  \bibinfo{author}{Mark~P. \surnamestart Jones\surnameend}
  (\bibinfo{year}{1996}): \bibinfo{title}{A Polymorphic Type System for
  Extensible Records and Variants}.
\newblock \bibinfo{type}{Technical Report}, \bibinfo{institution}{University of
  Nottingham}.

\bibitemdeclare{article}{HarperM93}
\bibitem{HarperM93}
\bibinfo{author}{Robert \surnamestart Harper\surnameend} \&
  \bibinfo{author}{John~C. \surnamestart Mitchell\surnameend}
  (\bibinfo{year}{1993}): \bibinfo{title}{On the Type Structure of
  Standard {ML}}.
\newblock {\sl \bibinfo{journal}{{ACM} Trans. Program. Lang. Syst.}}
  \bibinfo{volume}{15}(\bibinfo{number}{2}), pp. \bibinfo{pages}{211--252},
  \doi{10.1145/169701.169696}.

\bibitemdeclare{inproceedings}{HarperP91}
\bibitem{HarperP91}
\bibinfo{author}{Robert \surnamestart Harper\surnameend} \&
  \bibinfo{author}{Benjamin~C. \surnamestart Pierce\surnameend}
  (\bibinfo{year}{1991}): \bibinfo{title}{A Record Calculus Based on
  Symmetric Concatenation}.
\newblock In \bibinfo{editor}{David~S. \surnamestart Wise\surnameend}, editor:
  {\sl \bibinfo{booktitle}{Conference Record of the Eighteenth Annual {ACM}
  Symposium on Principles of Programming Languages, Orlando, Florida, USA,
  January 21-23, 1991}}, \bibinfo{publisher}{{ACM} Press}, pp.
  \bibinfo{pages}{131--142}, \doi{10.1145/99583.99603}.

\bibitemdeclare{article}{Jategaonkar93}
\bibitem{Jategaonkar93}
\bibinfo{author}{Lalita~A. \surnamestart Jategaonkar\surnameend} \&
  \bibinfo{author}{John~C. \surnamestart Mitchell\surnameend}
  (\bibinfo{year}{1993}): \bibinfo{title}{Type Inference with Extended
  Pattern Matching and Subtypes}.
\newblock {\sl \bibinfo{journal}{Fundam. Inf.}}
  \bibinfo{volume}{19}(\bibinfo{number}{1–2}), p. \bibinfo{pages}{127–165},
  \doi{10.3233/FI-1993-191-206}.

\bibitemdeclare{book}{Jones94a}
\bibitem{Jones94a}
\bibinfo{author}{Mark~P. \surnamestart Jones\surnameend}
  (\bibinfo{year}{1994}): \bibinfo{title}{Qualified Types: Theory and
  Practice}.
\newblock \bibinfo{series}{Distinguished Dissertations in Computer Science},
  \bibinfo{publisher}{Cambridge University Press},
  \doi{10.1017/CBO9780511663086}.

\bibitemdeclare{article}{Jones94}
\bibitem{Jones94}
\bibinfo{author}{Mark~P. \surnamestart Jones\surnameend}
  (\bibinfo{year}{1994}): \bibinfo{title}{A Theory of Qualified Types}.
\newblock {\sl \bibinfo{journal}{Sci. Comput. Program.}}
  \bibinfo{volume}{22}(\bibinfo{number}{3}), pp. \bibinfo{pages}{231--256},
  \doi{10.1016/0167-6423(94)00005-0}.

\bibitemdeclare{inproceedings}{Leijen05}
\bibitem{Leijen05}
\bibinfo{author}{Daan \surnamestart Leijen\surnameend} (\bibinfo{year}{2005}):
  \bibinfo{title}{Extensible records with scoped labels}.
\newblock In \bibinfo{editor}{Marko C. J.~D. \surnamestart van
  Eekelen\surnameend}, editor: {\sl \bibinfo{booktitle}{Revised Selected Papers
  from the Sixth Symposium on Trends in Functional Programming, {TFP} 2005,
  Tallinn, Estonia, 23-24 September 2005}}, {\sl \bibinfo{series}{Trends in
  Functional Programming}}~\bibinfo{volume}{6}, \bibinfo{publisher}{Intellect},
  pp. \bibinfo{pages}{179--194}.

\bibitemdeclare{article}{Ohori95}
\bibitem{Ohori95}
\bibinfo{author}{Atsushi \surnamestart Ohori\surnameend}
  (\bibinfo{year}{1995}): \bibinfo{title}{A Polymorphic Record Calculus
  and Its Compilation}.
\newblock {\sl \bibinfo{journal}{{ACM} Trans. Program. Lang. Syst.}}
  \bibinfo{volume}{17}(\bibinfo{number}{6}), pp. \bibinfo{pages}{844--895},
  \doi{10.1145/218570.218572}.

\bibitemdeclare{inproceedings}{OhoriB88}
\bibitem{OhoriB88}
\bibinfo{author}{Atsushi \surnamestart Ohori\surnameend} \&
  \bibinfo{author}{Peter \surnamestart Buneman\surnameend}
  (\bibinfo{year}{1988}): \bibinfo{title}{Type Inference in a Database
  Programming Language}.
\newblock In: {\sl \bibinfo{booktitle}{Proceedings of the 1988 {ACM} Conference
  on {LISP} and Functional Programming, {LFP} 1988, July 25-27, 1988, Snowbird,
  Utah, {USA.}}}, \bibinfo{publisher}{{ACM}}, pp. \bibinfo{pages}{174--183},
  \doi{10.1145/62678.62700}.

\bibitemdeclare{article}{PT1994}
\bibitem{PT1994}
\bibinfo{author}{Benjamin~C. \surnamestart Pierce\surnameend} \&
  \bibinfo{author}{David~N. \surnamestart Turner\surnameend}
  (\bibinfo{year}{1994}): \bibinfo{title}{Simple type-theoretic
  foundations for object-oriented programming}.
\newblock {\sl \bibinfo{journal}{Journal of Functional Programming}}
  \bibinfo{volume}{4}(\bibinfo{number}{2}), p. \bibinfo{pages}{207–247},
  \doi{10.1017/S0956796800001040}.

\bibitemdeclare{inproceedings}{Remy89}
\bibitem{Remy89}
\bibinfo{author}{Didier \surnamestart R{\'{e}}my\surnameend}
  (\bibinfo{year}{1989}): \bibinfo{title}{Typechecking Records and
  Variants in a Natural Extension of {ML}}.
\newblock In: {\sl \bibinfo{booktitle}{Conference Record of the Sixteenth
  Annual {ACM} Symposium on Principles of Programming Languages, Austin, Texas,
  USA, January 11-13, 1989}}, \bibinfo{publisher}{{ACM} Press}, pp.
  \bibinfo{pages}{77--88}, \doi{10.1145/75277.75284}.

\bibitemdeclare{inproceedings}{Remy92}
\bibitem{Remy92}
\bibinfo{author}{Didier \surnamestart R{\'{e}}my\surnameend}
  (\bibinfo{year}{1992}): \bibinfo{title}{Typing Record Concatenation for
  Free}.
\newblock In \bibinfo{editor}{Ravi \surnamestart Sethi\surnameend}, editor:
  {\sl \bibinfo{booktitle}{Conference Record of the Nineteenth Annual {ACM}
  {SIGPLAN-SIGACT} Symposium on Principles of Programming Languages,
  Albuquerque, New Mexico, USA, January 19-22, 1992}},
  \bibinfo{publisher}{{ACM} Press}, pp. \bibinfo{pages}{166--176},
  \doi{10.1145/143165.143202}.

\bibitemdeclare{inbook}{Remy94}
\bibitem{Remy94}
\bibinfo{author}{Didier \surnamestart R\'{e}my\surnameend}
  (\bibinfo{year}{1994}): 
  \bibinfo{title}{Type Inference for Records in
  Natural Extension of ML}, p. \bibinfo{pages}{67–95}.
\newblock \bibinfo{publisher}{MIT Press}, \bibinfo{address}{Cambridge, MA,
  USA}.

\bibitemdeclare{inproceedings}{Wand87}
\bibitem{Wand87}
\bibinfo{author}{Mitchell \surnamestart Wand\surnameend}
  (\bibinfo{year}{1987}): \bibinfo{title}{Complete Type Inference for
  Simple Objects}.
\newblock In: {\sl \bibinfo{booktitle}{Proceedings of the Symposium on Logic in
  Computer Science {(LICS} '87), Ithaca, New York, USA, June 22-25, 1987}},
  \bibinfo{publisher}{{IEEE} Computer Society}, pp. \bibinfo{pages}{37--44}.
\newblock
  \urlprefix\url{http://www.ccs.neu.edu/home/wand/papers/wand-lics-87.pdf}.

\bibitemdeclare{inproceedings}{Wand89}
\bibitem{Wand89}
\bibinfo{author}{Mitchell \surnamestart Wand\surnameend}
  (\bibinfo{year}{1989}): \bibinfo{title}{Type Inference for Record
  Concatenation and Multiple Inheritance}.
\newblock In: {\sl \bibinfo{booktitle}{Proceedings of the Fourth Annual
  Symposium on Logic in Computer Science {(LICS} '89), Pacific Grove,
  California, USA, June 5-8, 1989}}, \bibinfo{publisher}{{IEEE} Computer
  Society}, pp. \bibinfo{pages}{92--97}, \doi{10.1109/LICS.1989.39162}.

\end{thebibliography}
\end{document}